\documentclass[acmsmall,nonacm]{acmart/acmart}
\usepackage{graphicx} 
\usepackage{amsfonts}
\usepackage{amsmath}
\usepackage{thm-restate}
\usepackage{float}
\usepackage{cleveref}

\usepackage{algorithm}
\usepackage{algpseudocode}
\makeatletter\newcommand{\algmargin}{\the\ALG@thistlm}
\algdef{SE}[UPON]{Upon}{EndUpon}[3]{\textbf{upon} \textbf{#1} $\langle \text{#2} \mid #3\rangle$ \textbf{do}}{}
\algdef{SE}[UPONSIMPLE]{UponSimple}{EndUpon}[1]{\textbf{upon} #1 \textbf{do}}{}
\algdef{SE}[InParallel]{InParallel}{EndParallel}[1]{\textbf{in parallel} #1 \textbf{do}}{}
\algdef{SE}[UPONEXISTS]{UponExists}{EndUpon}[1]{\textbf{upon exists} #1 \textbf{do}}{}
\algdef{SE}[UPONTRUE]{UponTrue}{EndUpon}[1]{\textbf{upon once} #1 \textbf{do}}{}
\algdef{SE}[USES]{Uses}{EndUses}{\textbf{Uses:}}{}
\ifthenelse{\equal{\ALG@noend}{t}}%
  {\algtext*{EndUses}}
  {}%
\ifthenelse{\equal{\ALG@noend}{t}}%
  {\algtext*{EndUpon}}
  {}%
\algtext*{EndUpon}
\makeatother
\algnewcommand{\parState}[1]{
    \parbox[t]{\dimexpr\linewidth-\algmargin}{\strut\hangindent=\algorithmicindent \hangafter=1 #1\strut}}
\algblock{As}{EndAs}
\algnewcommand\algorithmicas{\textbf{as}}
\algrenewtext{As}[1]{\algorithmicas\ #1}
\algtext*{EndAs}

\include{macros}

\usepackage{todonotes}

\newtheorem{property}{Property}



\author{Marc Dufay}
\orcid{0009-0005-8440-8007}
\affiliation{%
  \institution{ETH Zurich}
   \country{Switzerland}
}
\email{mdufay@ethz.ch}

\author{Anton Paramonov}
\orcid{0009-0000-0760-8746}
\affiliation{%
  \institution{ETH Zurich}
   \country{Switzerland}
}
\email{aparamonov@ethz.ch}

\author{Roger Wattenhofer}
\orcid{0000-0002-6339-3134}
\affiliation{%
  \institution{ETH Zurich}
   \country{Switzerland}
}
\email{wattenhofer@ethz.ch}

\date{}

\begin{document}

\title{Optimal Adaptive Multi-Valued Byzantine Agreement}

\begin{abstract}
    In Byzantine Agreement (BA), $n$ parties, out of which $t$ can be Byzantine, run a distributed protocol to agree on a common valid input. Traditionally, these protocols have a linear latency and quadratic message complexity, making them impractical at a large scale. In their recent work, Constantinescu, Dufay, Paramonov, and Wattenhofer consider the actual number of byzantine parties $f \leq t$ and work toward decoupling the dependency on $n$ and $t$ in the complexity. They obtain a BA protocol with $\tilde{\mathcal{O}}(n + t\cdot f)$ message complexity and $\tilde{\mathcal{O}}(f)$ round complexity. 
    
    However, their results are strictly limited to agreement on a binary value. Using the framework given by their work along with novel techniques, we extend these results for BA on an $L$-bit value. With $\kappa$ being a security parameter, and with optimal resiliency ($t < n/2$ in the synchronous setting or $t < n/3$ otherwise), we obtain:
    \begin{itemize}
        \item In synchrony, a deterministic protocol with $\mathcal{O}(n\cdot (L + f \cdot \kappa ))$ bit complexity and $\mathcal{O}(f + \log n)$ round complexity.
        \item In synchrony and partial synchrony, deterministic protocols with $\tilde{\mathcal{O}}(n \cdot \kappa +  t\cdot (L + f \cdot \kappa))$ bit complexity and $\mathcal{O}(f)$ round complexity.
        \item In asynchrony, a protocol with $\tilde{\mathcal{O}}(n \cdot \kappa +  t\cdot(L + t \cdot \kappa))$ expected bit complexity and expected $\mathcal{O}(1)$ latency.
    \end{itemize}
\end{abstract}

\maketitle

\section{Introduction}
    In Byzantine Agreement, $n$ processes must reach agreement despite up to $t$ of them behaving maliciously. Soon after the problem was formalized by Lamport, Shostak, and Pease~\cite{lamport1982byzantine}, several fundamental lower bounds were established. For example, Dolev and Reischuk~\cite{dolev1985bounds} showed that any deterministic Byzantine agreement protocol requires $\Omega(t^2)$ messages in the worst case, and Dolev and Strong~\cite{dolev1990early} proved a worst-case lower bound of at least $t+1$ rounds.

    Since then, substantial research effort has been devoted to improving scalability despite these worst-case limitations. Prominent approaches include relaxing the problem specification~\cite{king2006towards,chaudhuri2000tight,dolev1982byzantine,feldman1997optimal}, tolerating probabilistic error~\cite{rabin1983randomized,cachin2000random,cohen2020not,blum2020asynchronous}, strengthening the model via additional trust assumptions~\cite{ben2025byzantine,liu2018scalable}, and optimizing for executions under optimistic conditions~\cite{kursawe2002optimistic,song2008bosco}. Within this latter landscape, adaptive protocols—i.e., protocols whose cost depends on the actual number of faults $f$ encountered in a given execution—have proven particularly effective in modern BFT system designs~\cite{yin2019hotstuff,civit2024dare}.
    
    More recently, a complementary direction has been proposed, motivated by the scale of modern deployments. Constantinescu, Dufay, Paramonov, and Wattenhofer \cite{constantinescu2025few} argue that contemporary systems may consist of thousands of nodes, making the classical assumption that a constant fraction of the system (e.g., $t \approx n/3$) can be corrupted overly conservative. This perspective motivates the design of fast and lightweight consensus protocols for the large-scale regime in which $t \ll n$.
    
    The result of~\cite{constantinescu2025few} initiates this line of work by presenting (almost) optimal protocols for binary agreement. In this paper, we further develop adaptive large-scale algorithms and present near-optimal protocols that support agreement over multi-valued domains. While we share a high-level approach with~\cite{constantinescu2025few}, we also introduce new techniques that may be of independent interest. 
\subsection{Our results}
    We give almost optimal adaptive large-scale algorithms for all major time models, that is, synchrony, partial synchrony, and asynchrony.

    We start with our main result for a synchronous network:

    \begin{theorem}
        Consider a PKI setup, a computationally bounded adversary, and a synchronous network. There exists a deterministic protocol that solves Multi-Valued Byzantine Agreement in this setting that tolerates up to $t < n/2$ faults and in a run with $f \leq t$ faults has a bit complexity of $O(n \cdot \log n \cdot \log t\cdot \kappa + t\cdot  (L +\log t \cdot \kappa) \cdot \log t+ t\cdot f \cdot \kappa )$, where $L$ is the input value bit-length, and round complexity of $O(f)$.
    \end{theorem}

    For partial synchrony with perfect clocks, we give the following result:
    
    \begin{theorem}
        Consider a PKI setup, a computationally bounded adversary, and a partially synchronous network. There exists a deterministic protocol that solves Multi-Valued Byzantine Agreement in this setting that tolerates up to $t < n/3$ faults and in a run with $f \leq t$ faults has a bit complexity of $O(n \cdot \log n \cdot \log t\cdot \kappa + t\cdot  (L +\log t \cdot \kappa) \cdot \log t+ t\cdot f \cdot \kappa )$, where $L$ is the input value bit-length, and round complexity of $O(f)$.
    \end{theorem}

    In both theorems, resilience is optimal, and round complexity is asymptotically optimal. For bit complexity, the lower bound is $\Omega(n +  t \cdot L  + t \cdot f)$, hence in both cases the bit complexity is optimal up to $\kappa$ and log factors.

    For asynchrony, we provide the following result:

    \begin{theorem}
        Consider a PKI setup, a computationally bounded adversary, and an asynchronous network. There exists a randomized protocol that solves Multi-Valued Byzantine Agreement with probability $1$ in this setting that tolerates up to $t < n/3$ faults and has expected bit complexity of $O(n\cdot\log n\cdot \kappa + t\cdot L + t^2\cdot \kappa)$, where $L$ is the input value bit-length, and expected round complexity of $O(1)$.
    \end{theorem}

    In addition to large-scale protocols, we provide an adaptive synchronous protocol which avoids log factors in bit complexity in case $t$ is of the order of $n$.
    \begin{theorem}
        Consider a PKI setup, a computationally bounded adversary, and a synchronous network. There exists a deterministic protocol that solves Multi-Valued Byzantine Agreement in this setting that tolerates up to $t < n/2$ faults and in a run with $f \leq t$ faults has a bit complexity of $O(n \cdot L + n\cdot f\cdot \kappa)$, where $L$ is the input value bit-length, and round complexity of $O(f + \log n)$.
    \end{theorem}

\subsection{High-Level approach}
    Our overall protocols rely on $3$ different components which are conceptually run one after the other. They all rely on a \textbf{quorum} of $\Theta(t)$ parties, inspired by \cite{constantinescu2025few}, running the actual agreement protocol:

\paragraph{All to Quorum Broadcast (AQB)} when $t \ll n$, we remark that all existing MVBA protocols still have bit complexity $\Omega(nL)$. However, we show this is not a lower bound for this problem. More specifically, our protocols have bit complexity $\tilde{\mathcal{O}}(tL)$ when considering only the dependency in $L$. This is done thanks to the AQB phase, where all $n$ parties send hashes of their input in a sparse way to the quorum. This step ensures that if value gets decided, at most $\mathcal{O}(t)$ parties do not have it as their input and in turn allows us to reduce the bit complexity to $\tilde{\mathcal{O}}(tL + poly(n,t,f))$. We remark that this step is completely new compared to \cite{constantinescu2025few}, as it is not required for binary agreement when $L = 1$.

\paragraph{Quorum Agreement} All parties in the quorum run an agreement to decide a value. Note that the properties we require of the agreement protocol are somewhat unusual. First, we only require at least $n/4$ honest parties to decide a value within $\mathcal{O}(f)$ rounds. The remaining parties may decide later or not even decide within the quorum agreement. We call this property partial termination. Moreover, we require the decided value to come with a \emph{certificate}, i.e a proof that any party can use to confirm this value was decided. We note that many additional new techniques are required to ensure validity and agreement for an $L$-bit value compared to a binary value.

\paragraph{Quorum to All Broadcast (QAB)} Once a party in the quorum has decided, it will try to send its output value to all $n$ parties in a sparse way. The main differences with the implementation from \cite{constantinescu2025few} is that (1) error coding is used to ensure a bit-efficient transmission of the $L$-bit decided value and (2) a hash check is performed prior to sending the value to parties to guarantee the $\tilde{\mathcal{O}}(tL + poly(n,t,f))$ bit complexity we set up the AQB for.

\smallskip

When $t = \Theta(n)$, one can run the Quorum Agreement protocol on its own to save $\log$ factors in the bit complexity. However, we remark that our Quorum Agreement protocol for partial synchrony on its own only satisfies partial termination (only guarantees that a fourth of the parties decide) and thus cannot be used for regular Byzantine Agreement. Meanwhile, our synchronous Quorum Agreement protocol satisfies partial termination in $\mathcal{O}(f)$ rounds and termination in $\mathcal{O}(f + \log n)$. So it can be used to solve MVBA with optimal resiliency $t < n/2$, bit complexity $\mathcal{O}(n \cdot (L + f \cdot \kappa))$ and near optimal round complexity $\mathcal{O}(f + \log n)$.
    
\section{Related Work}

\paragraph*{Synchrony.}
Reducing the communication cost of Byzantine Agreement (BA) has been a central theme in distributed computing, especially in fully synchronous networks.

In this setting, Cohen et al.~\cite{cohen2022brief} give deterministic protocols for Weak Byzantine Agreement and Byzantine Broadcast with bit complexity $\mathcal{O}(n\cdot f\cdot\kappa)$ where $\kappa$ is a security parameter. Shortly after, Civit et al.~\cite{civit2023strong} design a deterministic binary BA protocol that also communicates $\mathcal{O}(n\cdot f\cdot \kappa)$ bits and tolerates $t < \left(\frac{1}{2}-\varepsilon\right)n$ Byzantine faults.

This research direction culminates in Civit et al.~\cite{civit2024dare}, which extends beyond binary BA to Multi-Valued Byzantine Agreement (MVBA), where processes may propose arbitrary values (not only $0$ and $1$). The same work also addresses Interactive Consistency (outputting a vector containing the proposals of the honest processes) and supports External Validity, requiring that the decided value satisfy a fixed predicate. The resulting bit complexity is
$\mathcal{O}(Ln + n(f+1)\kappa)$,
where $L$ denotes the bit-length of the input (e.g., $L=1$ for binary BA). A notable aspect of~\cite{civit2024dare} is its reliance on the comparatively heavy Multiverse Threshold Signature Scheme (MTSS)~\cite{baird2023threshold,garg2024hints}. It explicitly asks whether adaptive BA can be obtained using lighter cryptographic tools; the present work answers this question in the affirmative.

Despite these communication improvements, the above synchronous protocols still have worst-case round complexity $\Omega(n)$, even when the actual number of faults $f$ is small. A complementary line of work studies \emph{early-stopping} protocols. For instance, Lenzen and Sheikholeslami~\cite{lenzen2022recursive} solve binary BA in $\mathcal{O}(f)$ rounds, but with bit complexity $\mathcal{O}(nt)$ and resilience restricted to $t<n/3$. We note that early-stopping protocols enforce stronger properties in that, after $\mathcal{O}(f)$ rounds, parties exit the protocol and do not send any messages anymore. Meanwhile, we allow our protocols to send messages, even after deciding.

The immediate predecessor of the present work~\cite{constantinescu2025few} achieves both early stopping and adaptive communication. Concretely, it provides two synchronous binary BA protocols: one with $O(f)$ rounds and $O(nf\kappa)$ bits of communication, and a second with $O(f)$ rounds and $O\big((n\log t + tf)\log n\cdot\kappa\big)$ bits.

Lower bounds delineate the limits of these results. The classical theorem of Dolev and Reischuk~\cite{dolev1985bounds} shows that any synchronous BA protocol must use $\Omega(f^2)$ messages. Spiegelman~\cite{spiegelman2020search} strengthens this to $\Omega(n+tf)$. Moreover, Abraham et al.~\cite{abraham2019communication} prove that even randomized synchronous protocols cannot avoid the $\Omega(f^2)$ message lower bound against a highly adaptive adversary that can remove messages after the fact. Regarding time complexity, Dolev and Strong~\cite{dolev1990early} establish that achieving Byzantine agreement requires at least $\min\{f+2,\,t+1\}$ rounds. Finally, for Multi-Valued BA, the folklore knowledge proclaims that at least $\Omega(L\cdot t)$ bits are needed. The proof sketch for this can be found in \cite{fitzi2006optimally}.

\paragraph*{Partial Synchrony.}
Byzantine Agreement under partial synchrony has also seen major progress. Civit et al.~\cite{civit2024byzantine} give the first partially synchronous binary BA protocol with $\mathcal{O}(n^2\cdot \kappa)$ bit complexity. Subsequently, Civit et al.~\cite{civit2025partial} introduce the \emph{Oper} framework, which compiles any synchronous BA protocol into a partially synchronous one while preserving the worst-case \emph{per-process} bit complexity. However, when instantiated with protocols such as~\cite{civit2024dare}, Oper yields only a total $\mathcal{O}(n^2)$ communication bound (rather than adaptive $\mathcal{O}(n\cdot f)$), because the underlying synchronous protocol already has $\Omega(n)$ per-process bit complexity.

The worst-case round complexities of~\cite{civit2024byzantine} and~\cite{civit2024dare} are $\Theta(t)$ and $\Theta(n)$, respectively.

For binary BA,~\cite{constantinescu2025few} provides partially synchronous protocols that are both adaptive and early-stopping, and are close to optimal. Specifically, it presents (i) a protocol running in $O(f)$ rounds with $O(nf\cdot\kappa)$ communication, and (ii) a protocol running in $O(f\log n)$ rounds with $O\big((n+t f)\log n \log t\cdot\kappa\big)$ communication.

Beyond BA, State Machine Replication (SMR) is tightly connected. SMR requires External Validity, meaning the decided value satisfies some predefined predicate. A well-known partially synchronous SMR protocol is HotStuff~\cite{yin2019hotstuff}. HotStuff proceeds in \emph{views} led by a designated leader, and termination is obtained once the honest parties remain in the same view long enough. The cost of synchronizing views depends strongly on the underlying clock assumptions. With perfectly synchronized clocks, one may define views as fixed time intervals, essentially eliminating view-synchronization overhead. In more fine-grained models, however, clocks need not be perfectly aligned. While HotStuff does not formally fix a clock model, it uses exponential back-off to handle discrepancies; this is theoretically justified but can be inefficient in practice. Recent work on view synchronization~\cite{lewis-pye2023optimal, lewis2022quadratic, lewis2023fever} provides a more robust mechanism that guarantees consistent views from only partial initial clock synchronization, achieving $\mathcal{O}(n\cdot f\cdot\kappa)$ bits of communication. When combined with HotStuff, this yields the same overall $\mathcal{O}(n\cdot f\cdot \kappa)$ communication bound in partial synchrony. 

Finally, Spiegelman et al.~\cite{spiegelman2020search} show that in partially synchronous networks, if one counts messages sent before the Global Stabilization Time (GST), then any protocol may require an unbounded number of messages. Accordingly, communication complexity for partially synchronous protocols is stated only for the post-GST period.

\paragraph*{Asynchrony.} The well-known FLP impossibility result \cite{fischer1985impossibility} establishes that Byzantine Agreement cannot be deterministically solved in asynchronous networks if even a single process can fail. However, the use of randomness provides a way to circumvent this limitation. 

For example, Cachin et al. \cite{cachin2000random} present a protocol that achieves strong BA in asynchrony, ensuring Agreement and Validity unconditionally while guaranteeing Termination with probability $1$. For binary inputs, this protocol operates with a bit complexity of $\mathcal{O}(n^2\kappa)$. 

Further improving on message complexity while allowing a negligible probability of failure, Cohen et al. \cite{cohen2020not} propose a solution for binary BA in asynchrony. Their protocol achieves Agreement, Validity, and Termination with high probability (w.h.p.) and has an expected bit complexity of $\mathcal{O}(n \cdot \log^2 n\cdot\kappa)$. The predecessor of the present paper \cite{constantinescu2025few} has given a binary algorithm that succeeds with probability $1$ and has a bit complexity of $O((n+t^2)\cdot \log n\cdot \kappa)$.

For Multi Value Byzantine Agreement in Asynchrony, one of the state-of-the-art results is the Dumbo-MVBA protocol \cite{yuhan2020dumbo} which possesses an expected bit complexity of $O(L\cdot n+n^2 \cdot \kappa)$ and expected round complexity of $O(1)$.

We want to highlight that, unlike the result of \cite{cohen2020not}, our work and \cite{yuhan2020dumbo} focus on the asynchronous algorithms with the strongest probabilistic guarantees, i.e., on algorithms that succeed with probability $1$ and against an adaptive adversary.

All these protocols enjoy constant expected latency. 

On the impossibility side, \cite{constantinescu2025few} shows that for binary BA in asynchrony, at least $\Omega(n + t^2)$ bits of communication is needed.

Table~\ref{tbl:related work} summarizes prior results as well as the contribution of the present work.

\begin{table}[ht]
\resizebox{\textwidth}{!}{%
\begin{tabular}{|l|l|l|l|l|l|}
\hline
\textbf{Paper} & \textbf{Problem}  & \textbf{Bit Complexity} & \textbf{Rounds} & \textbf{Resilience} & \textbf{Crypto} \\
\hline
\multicolumn{6}{|c|}{\textbf{Synchrony}} \\ 
\hline

\cite{nayak2020extension} & MVBA & $\mathcal{O}(n\cdot L + n^2\cdot \kappa)$ & $\mathcal{O}(n)$ & $t < n/2$ & CA\\

\cite{civit2024dare} & MVBA, IC & $\mathcal{O}(n\cdot L + n\cdot f\cdot \kappa)$ & $\mathcal{O}(n)$ & $t < n/2$ & MTSS, CA\\

\cite{lenzen2022recursive} & BA & $\mathcal{O}(n\cdot t)$ & $\mathcal{O}(f)$ & $t < n/3$ &  None \\

\cite{constantinescu2025few} & BA & $\mathcal{O}(n \cdot f \cdot \kappa)$ & $\mathcal{O}(f)$ & $t < n/2$ & T-Sig \\

\cite{constantinescu2025few} & BA & $\mathcal{O}((n + t\cdot f)\cdot \log^2 n \cdot \kappa)$ & $\mathcal{O}(f)$ & $t < n/2$ & T-Sig  \\

\textbf{This paper} & MVBA  & $\mathcal{O}(n \cdot L + n\cdot f \cdot \kappa)$ & $\mathcal{O}(f + \log n)$ & $t < n/2$ & T-Sig, CA  \\

\textbf{This paper} & MVBA  & $\mathcal{O}((n +  t \cdot L  + t \cdot f)\cdot \log^2 n \cdot \kappa )$ & $\mathcal{O}(f)$ & $t < n/2$ & T-Sig, CA  \\

\cite{spiegelman2020search} & BA & $\Omega(n + t\cdot f)$ & Any & Any & Any  \\

Folklore & MVBA & $\Omega(n + t\cdot L)$ & Any & Any & Any \\

\hline
\multicolumn{6}{|c|}{\textbf{Partial Synchrony}} \\ 
\hline

\cite{yin2019hotstuff, lewis2023fever}$^\ast$ & SMR  & $\mathcal{O}(n \cdot f\cdot \kappa)$ & $\mathcal{O}(f)$ & $t < n/3$ & T-Sig, Hash\\

\cite{civit2025partial} & MVBA &  $\mathcal{O}(n\cdot L + n^2\cdot \log n\cdot \kappa)$ & $\mathcal{O}(n)$ & $t < n/3$ & Hash \\

\cite{civit2025partial} & MVBA &  $\mathcal{O}(n\cdot L + n^2\cdot \log n)$ & $\mathcal{O}(n)$ & $t < n/5$ & None \\

\cite{constantinescu2025few} & BA & $\mathcal{O}(n \cdot f \cdot \kappa)$ & $\mathcal{O}(f)$ & $t < n/3$ & T-Sig\\

\cite{constantinescu2025few} & BA & $\mathcal{O}((n + t \cdot f)\cdot \log^2 n \cdot  \kappa)$ & $\mathcal{O}(f\cdot \log n)$ & $t < n/3$ & T-Sig\\

\textbf{This paper} & MVBA  & $\mathcal{O}((n+t\cdot L + t\cdot f)\cdot \log^2 n  \cdot \kappa)$ & $\mathcal{O}(f)$ & $t < n/3$ & T-Sig, CA  \\

\cite{spiegelman2020search} & BA & $\Omega(\infty)^\S$ & Any & Any & Any \\

\hline
\multicolumn{6}{|c|}{\textbf{Asynchrony}} \\ 
\hline

\cite{cachin2000random} & BA & $\mathbb{E}[\mathcal{O}(n^2\cdot \kappa)]$ & $\mathbb{E}[O(1)]$ &  $t < n/3$ & T-Sig,Coin  \\

\cite{cohen2020not} & BA & $\mathbb{E}[\mathcal{O}(n \cdot \log^2 n \cdot \kappa)]$ & $\mathbb{E}[O(1)]$ & $t < (\frac{1}{3} - \varepsilon)n$ & VRF \cite{micali1999verifiable} \\

\cite{constantinescu2025few} & BA & $\mathbb{E}[\mathcal{O}((n + t^2)\cdot \log n\cdot\kappa )]$ & $\mathbb{E}[O(1)]$ & $t < n/3$ & T-Sig,Coin\\ 

\cite{yuhan2020dumbo} & MVBA & $\mathbb{E}[\mathcal{O}(L\cdot n+n^2 \cdot (\log n +\kappa))]$ & $\mathbb{E}[O(1)]$ & $t < n/3$ & T-Sig,Coin\\

\textbf{This paper} & MVBA  & $\mathbb{E}[\mathcal{O}((n\cdot \kappa   + t\cdot L + t^2\cdot \kappa) \cdot \log n )]$ & $\mathbb{E}[O(1)]$ & $t < n/3$ & T-Sig, Coin, CA  \\

\cite{constantinescu2025few} & BA & $\mathbb{E}[\Omega(n + t^2)]$ & Any & Any & Any\\

\hline

\end{tabular}
}
\vspace{0.5cm}
\caption{Comparison of Byzantine Agreement Protocols. When a communication is stated with $\Omega$, it indicates a lower bound result. \\
$\ast$ -  \cite{lewis2023fever} uses random leader election, but can be derandomized, maintaining its characteristics.\\
$L$ - size of the input in the MVBA. $\kappa$ - security parameter for cryptographic primitives.\\
$\mathbb{E}$ - in expectation, $\S$ - unlimited messages before GST, MTSS - Multiverse Threshold Signatures Scheme \cite{lee2023decentralized}, T-Sig - Threshold Signatures Scheme \cite{boneh2004short}, CA - Cryptographic Accumulator.
}
\label{tbl:related work}
\end{table}
\paragraph*{On the use of communication graphs.}
Even when the underlying network allows all-to-all messaging, it is often beneficial to \emph{limit} who talks to whom by imposing a \emph{communication graph}, thereby controlling the total communication volume. Although this idea is natural and has shown promise, it has not yet been studied extensively in the context of distributed agreement problems.

This paper adopts the approach of~\cite{constantinescu2025few}, the first work to leverage bipartite dispersers for Byzantine Agreement (BA). Among other related efforts, Elsheimy et al.~\cite{elsheimy2024deterministic} and Civit et al.~\cite{civit2024dare} use expander-based communication patterns to obtain adaptive $\mathcal{O}(nf\cdot \kappa)$ bits of communication for BA; however, their techniques are fundamentally tailored to synchrony and, in partially synchronous networks, would lead to $\Omega(n^2)$ communication. We note that we re-use the approach from Civit et al. for our synchronous protocol. Moreover, communication graphs have also been employed for leader election, for example, using bipartite samplers and expanders as in~\cite{king2006scalable, bhangale2025leader}. These leader-election results, however, ensure agreement only for \emph{most} honest processes rather than \emph{all} honest processes. Finally, Chlebus et al.~\cite{chlebus2023deterministic} use expander graphs to spread a value initially known to a subset of nodes to the entire system, but their method applies only in synchronous settings and requires $t=\mathcal{O}(\sqrt{n})$.

\paragraph*{View Change and Clocks.} Most, if not all, major Partially Synchronous protocols starting with PBFT \cite{castro2002practical}, and continuing with Hotstuff \cite{yin2019hotstuff}, Algorand \cite{gilad2017algorand}, Tendermint \cite{buchman2018latest}, and  Simplex \cite{chan2023simplex} are view-based. A view is a logical time period during which a designated leader is responsible for driving the progress of the system. The challenge with this concept is to make sure that, at some point, all honest processes synchronize their views. This challenge goes hand in hand with assumptions about processors' clocks. For instance, if all clocks are perfectly synchronized, one could define view $i$ as the time interval $[4i\Delta,4(i+1)\Delta]$, making each view last $4\Delta$. If clocks have some bounded skew but do not drift (or drift slowly), one can do an exponential backoff as proposed in Hotstuff. An even weaker assumption, called partial initial clock synchronization \cite{lewis2023fever}, states that there exists a known $\Gamma$ such that at least $t+1$ honest parties start the protocol within $\Gamma$ time of the first honest start. 

Overall, to the best of our knowledge, the analysis of clock assumptions in partial synchrony and their effect on algorithms is lacking in the literature. The only result of this nature we are aware of is Fever \cite{lewis2023fever}. Fever implements view changes under partial initial clock synchronization so that before all honest parties remain in the same view with an honest leader for long enough, at most $\mathcal{O}((n+nf)\cdot \kappa)$, where $\kappa$ is a security parameter. Notably, Fever counts bits after $\mathrm{GST}+\Delta$ and not after GST, this being an ingrained assumption for their result.

\section{Preliminaries}
\subsection{Model and Definitions}
We study a distributed protocol executed by a set of $n$ parties $\mathcal{P} = \{p_1,\ldots,p_n\}$ connected via a fully connected network. Communication channels are authenticated: upon receiving a message, a party can identify its sender.

\noindent\textbf{Network Setting.} We consider three communication models:
\begin{itemize}
    \item \textit{Synchronous network.} There exists a known bound $\Delta > 0$ such that every message sent by an honest party is delivered within $\Delta$ time units. Consequently, protocols can be described in rounds of duration $\Delta$.
    \item \textit{Partially synchronous network.} We adopt the model of \cite{dwork1984consensus}. There is a known $\Delta>0$ and an unknown \emph{Global Stabilization Time} (GST) such that a message sent at time $t$ is delivered by time at most $\max\{t,\mathrm{GST}\}+\Delta$.
    \item Message delivery times are unbounded—there is no guaranteed limit on the delay between sending and receiving. The only assurance is eventual delivery: any message that is sent by an honest process will eventually be received.
\end{itemize}

\noindent\textbf{Byzantine Behavior.} We assume that up to $t$ parties may be corrupted and behave arbitrarily. An adversary controls all corrupted parties, and the protocol must tolerate such behavior. Parties that are not corrupted are called honest. We additionally denote by $f$ the \emph{actual} number of Byzantine parties in a particular execution, with $0 \le f \le t \le n$.

In addition, in the partially synchronous model prior to GST, the adversary is assumed to control the message scheduler: it may delay and reorder messages arbitrarily, subject only to delivery after GST.
Furthermore, in asynchrony, an adversary may delay and reorder messages arbitrarily, subject only to eventual delivery.

Our protocols tolerate a \emph{dynamic} adversary, meaning the adversary may adaptively choose which parties to corrupt after observing messages exchanged during the execution.

\noindent\textbf{Clocks and Views.} For our partially synchronous protocols, we use the standard notion of a \emph{view}, i.e., a logical epoch in which a designated leader is responsible for driving progress. Each view has an associated index, and all correct parties share a common leader-selection function $\mathrm{leader}(i)$ specifying the leader of view $i$. If the current leader does not make progress---because it crashes, equivocates, or the network delays its messages---parties trigger a view change and move to the next view with a new leader. Since parties may advance at different speeds, distinct honest parties may temporarily be in different views; parties then ignore leaders outside their current view, preventing those leaders from collecting enough votes/signatures/etc.\ to advance the protocol.

The view mechanism is required to satisfy the following liveness guarantee: \emph{after GST, all honest parties eventually remain in some view whose leader is honest for a sufficiently long period}. What constitutes ``sufficiently long'' depends on the protocol; in our partially synchronous constructions, this duration is always $\Theta(\Delta)$.

In our partially synchronous algorithms, we assume parties have perfectly synchronized clocks. Under this assumption, implementing views would be straightforward without any communication: for example, one could define view $i$ as the time interval $[4i\Delta,4(i+1)\Delta]$, making each view last $4\Delta$.

We leave determining the minimal clock assumptions for implementing adaptive partially synchronous protocols, for example, view synchronization, as an orthogonal open problem of high interest.

\noindent\textbf{Byzantine Agreement.} Our goal is to solve \emph{Strong Multi-Valued Binary Byzantine Agreement}; for brevity, we refer to it simply as Byzantine Agreement (BA). In BA, each honest party may \emph{propose} a value and eventually \emph{decide} a value. We consider different properties:
\begin{itemize}
    \item \textit{Termination}: every honest party eventually decides;
    \item \textit{Probabilistic Termination}: every honest party eventually decides eventually;
    \item \textit{Agreement}: if two honest parties decide $v_1$ and $v_2$, then $v_1=v_2$;
    \item \textit{Strong Unanimity}:\footnote{This property is also known as Strong Validity.} if all honest parties propose the same value, then that value is the one decided.
\end{itemize}
A protocol is said to solve Byzantine Agreement if it solves agreement, strong unanimity. In asynchrony, termination is replaced with probabilistic termination.

\noindent\textbf{Bit Complexity.} We define the \emph{bit complexity} of a protocol to be the total number of bits sent by honest parties.

As observed by Spiegelman et al.\ \cite{spiegelman2020search}, no partially synchronous BA protocol can guarantee a bounded number of messages if one counts communication occurring before GST. Accordingly, when we state bit complexity bounds for partially synchronous protocols, we measure only the communication sent \emph{after} GST.

\noindent\textbf{Round Complexity.} In the synchronous model, the round complexity is the first round by which all honest parties have decided. In the partially synchronous model, it is the first such round minus the round in which GST occurs. The round complexity of an asynchronous protocol is its round complexity assuming it runs in synchrony.

\subsection{Tools Used}

\paragraph*{Disperser}
\label{sec:disperser}
A central ingredient used throughout our algorithms is a \emph{disperser} graph.

\begin{definition}
    Consider a bipartite graph $G = (L \sqcup R, E)$ where each node on the left has degree $d$. $G$ is a $(k, \varepsilon)$ \emph{disperser} if for every $S \subseteq L$ of size at least $k$, it holds that $|N(S)| \geq (1 - \varepsilon)|R|$.

    Here $N(S)$ is a neighborhood of $S$, namely $N(S) = \{v \in R\mid (u, v) \in E\text{ for some } u\in S\}$.
\end{definition}

Intuitively, the disperser condition says that any sufficiently large subset of left vertices has neighbors covering almost all of the right side. Of course, this condition is easy to satisfy in uninteresting ways: for example, a complete bipartite graph is a disperser, but uses far too many edges. Similarly, if the threshold $k$ is close to $|L|$, then the guarantee may hold simply because $S$ is enormous. The main objective is therefore to obtain dispersers that simultaneously (i) have small left degree (ideally $d = O(\log n)$) and (ii) work already for subsets $S$ whose size is small relative to $|R|$.

Such graphs are known to exist. The following statement is standard (see, e.g., Theorem~1.10 in \cite{radhakrishnan2000bounds}).

\begin{theorem}
    \label{thm:bipartite disperser}
    For any $n, m$, $m \leq n$ and $\varepsilon > 0$, there exists a $(k, \varepsilon)$ bipartite disperser $G(L\sqcup R, E)$ with $|L| = n, |R| = m$, left degree $d = O_\varepsilon(\log n)$ and $k = \Omega_\varepsilon(\frac{m}{\log n})$. 
\end{theorem}

That said, the guarantee in Theorem~\ref{thm:bipartite disperser} is currently known only via a non-constructive existence proof. In contrast, explicit constructions are available (e.g., Ta-Shma et al.~\cite{ta2007lossless}), but they incur slightly weaker (polylogarithmic) parameters:

\begin{theorem}[Theorem 1.4~\cite{ta2007lossless}]
    \label{thm:explicit bipartite disperser}
    There exists a deterministic polynomial time algorithm that for any $n, m$, $m \leq n$ and $\varepsilon > 0$ outputs a $(k, \varepsilon)$ bipartite disperser $G(L\sqcup R, E)$ with $|L| = n, |R| = m$, left degree $d = O_\varepsilon(\mathrm{poly}(\log n))$ and $k = \Omega_\varepsilon(\frac{m\log ^3 n}{d})$. 
\end{theorem}

In what follows, we work with dispersers as in Theorem~\ref{thm:bipartite disperser}. Replacing them with the explicit construction of Theorem~\ref{thm:explicit bipartite disperser} would preserve correctness, but would introduce additional multiplicative polylogarithmic factors into our communication bounds.

We next highlight a key result that enables the use of bipartite dispersers for distributed protocols. Roughly speaking, it provides a way to map parties to committees so that—even if the adversary decides whom to corrupt only after learning the mapping—only a limited number of parties end up ``unlucky'' in the sense that all of their assigned committees contain Byzantine nodes.

We first recall the relevant terminology.

\begin{definition}[Definition 3.7 of \cite{constantinescu2025few}]
    We call a committee \emph{compromised} if it contains at least one byzantine party. 
    
    For a given assignment of parties to committees and a given set of byzantine parties, we say that a party is \emph{blocked} if all committees it is assigned to are compromised. 
\end{definition}

With this in place, the theorem is as follows.

\begin{theorem}[Theorem 3.8 of \cite{constantinescu2025few}]
    \label{thm:nodes2committees}
    Given $n$ parties and an integer $\hat{f}$, there is a way to assign parties to committees such that:
    \begin{enumerate}
        \item There are at most $O(\hat{f} \log n)$ committees.
        \item Every party is assigned to at most $O(\log n)$ committees.
        \item An adversary, after seeing the assignment, can make up to $\hat{f}$ arbitrary parties byzantine. Nonetheless, no matter the adversarial choice of which parties to corrupt, the number of blocked parties is at most $c_b\cdot \hat{f}$ for some constant $c_b$.
    \end{enumerate}
\end{theorem}

\paragraph*{Signatures} In some of our protocols, we will assume the presence of a public-key infrastructure (PKI). For a party $p$, the interface of PKI consists of functions $\textit{sign}_p$ and $\textit{verify}$ such that:\footnote{As is standard in cryptography, here and in similar contexts, we assume only party $p$ can invoke $\textit{sign}_p(\cdot)$, but all parties can invoke $\textit{verify}(\cdot, \cdot, p)$.}
\begin{itemize}
    \item For a given party $p$ and message $m$, $\textit{sign}_p(m)$ returns a signature $s$.
    \item For a given party $p$, value $s$ and message $m$, $\textit{verify}(m,s,p)$ returns $true$ if and only if $s = \textit{sign}_p(m)$.
\end{itemize}
Moreover, for some protocols, we will also assume the presence of a threshold signature scheme \cite{boneh2004short}. A threshold signature scheme is parameterized by two quantities: $k$ - the threshold, and $n$ - the number of parties participating in the scheme. Each party $p$ that participates in the scheme is given functions $\textit{tsign}_p$, $\textit{tcombine}$, and $\textit{tverify}$ such that:
\begin{itemize}
    \item For a given party $p$ and message $m$, $\textit{tsign}_p(m)$ returns a partial signature $\rho_p$.
    
    \item For given message $m$, a set $P \subseteq \mathcal{P}$ such that $|P| = k$ and tuple $(\rho_p)_{p \in P}$ such that $\forall p \in P,\ \rho_p = \textit{tsign}_p(m)$, a call $\textit{tcombine}\left(m, (\rho_p)_{p \in P}\right)$ returns a threshold signature $\sigma$.
    
    \item For a given message $m$ and threshold signature $\sigma$, $\textit{tverify}(m, \sigma)$ returns $true$ if and only if there exists a tuple of partial signatures $(\rho_p)_{p \in P}$ with $|P| = k$ such that $\sigma = \textit{tcombine}\left(m, (\rho_p)_{p \in P}\right)$.
\end{itemize}
Except when mentioned otherwise, the threshold signature will be used with parameter $k = n - t$ in this paper.

We will also use \emph{aggregate signature scheme} \cite{boneh2003aggregate}, which, for simplicity of presentation, we define to be a threshold signature scheme\footnote{The original definition is more general; it allows parties to sign different messages. We do not need this versatility in our results.} with $k = n.$ 

\paragraph*{Hash Functions}
A hash function maps inputs from a (typically large) domain $X$ to outputs in a (much smaller) range $Y$, where $|Y| \ll |X|$. For an adversary with bounded computational power, a cryptographic hash function $h$ is expected to satisfy the following property.

\begin{property}[Second-Preimage Resistance]
\label{prop:hash_second_preimage_resistance}
Given an input $x \in X$, it is computationally infeasible to find an $x' \in X$ with $x' \neq x$ such that $h(x') = h(x)$.
\end{property}


Security is governed by the output length of $h$. We therefore introduce a security parameter $\kappa$, defined so that a hash output can be encoded using $\kappa$ bits, i.e., $\kappa = \lceil \log_2 |Y| \rceil$. Larger $\kappa$ increases security, at the cost of increased communication in our protocols. In practice, $\kappa = 256$ bits is a commonly used security level. For a more detailed overview of hashing, see \cite{chi2017hashing}.

\paragraph*{Cryptographic Accumulator}
Our protocol relies on a cryptographic accumulator \cite{derler2015ca, kate2010ca} to share parts of the output value. A cryptographic accumulator is used to certify that a value belongs to a set. More specifically, given a set $S$, a cryptographic accumulator returns an accumulator and a witness for every element in the set.

Moreover, a function $\emph{verify}$ is given such that, given the accumulator $acc$, an element $x \in S$ and its witness $w$, $\emph{verify}(acc, x, w) = true$. It is computationally infeasible to find $y \notin S$ and $w'$ such that $\emph{verify}(acc, y, w') = true$. Both the accumulator and each witness have size $\mathcal{O}( \kappa)$ where $\kappa$ is a security parameter.

We note that optimal cryptographic accumulators require a setup and a known upper bound on the size of $S$. A cryptographic accumulator can also be implemented using Merkle trees \cite{merkle1987digital}. In this case, the witness size becomes $\mathcal{O}(\kappa \log |S|)$ but no setup is required.

\paragraph*{Erasure Coding}
Erasure coding expands an original message $m$ into a longer encoded message $m_{\mathsf{enc}}$ such that the original can be recovered from any sufficiently large subset of the encoded symbols. In particular, we use classical Reed--Solomon codes \cite{wicker1999reed}, which satisfy the following standard guarantee.

\begin{theorem}
Let $\mathbb{F}$ be a field with $N$ elements. Let $m$ be a message represented by $K$ elements of $\mathbb{F}$. Fix an integer $c$ such that $cK \le N$. Then there exists a Reed--Solomon encoding procedure that maps $m$ to an encoded message $m_{\mathsf{enc}}$ consisting of $cK$ elements of $\mathbb{F}$, with the property that $m$ can be reconstructed from any $K$ elements of $m_{\mathsf{enc}}$.
\end{theorem}

We will use the following corollary. From this point on, a \emph{symbol} refers to a string of bits.

\begin{corollary}
    \label{cor:reed-solomon}
Let $m$ be a message of length $L$ bits, and let $K,c$ be positive integers. There exists a procedure that produces $cK$ symbols such that $m$ can be reconstructed from any $K$ of them, and each symbol has size $O\!\left(\max\left\{\left\lceil \frac{L}{K}\right\rceil,\ \log(cK)\right\}\right)$ bits.
\end{corollary}

\paragraph*{Combining Erasure Coding and Cryptographic Accumulator}
In this paper, we always use a Cryptographic Accumulator together with Erasure Coding: Erasure Coding allows a large value to be divided into smaller chunks and a Cryptographic Accumulator is used to verify that these chunks come from the large value.

Therefore, we formalize this abstraction as two functions provided to our protocol:
\begin{itemize}
    \item \textbf{encode}$_n$($v$) takes as input an $L$-bit value $v$ and returns an accumulator $acc$ and $n$ shares $(s_1, ..., s_n)$ such that the accumulator has size $\mathcal{O}(\kappa)$ and each share contains both the erasure symbol and a witness for this symbol. A share has therefore size $\mathcal{O}(L/n + \log n + \kappa) = \mathcal{O}(L/n + \kappa)$.
    \item \textbf{decode}$_n$($(s_i)$) takes as input at least $n/4$ distinct shares $(s_i)$ which were obtained through a call to \textbf{encode}$_n$($v$) and returns $v$.
\end{itemize}

We note that checking the witness is done implicitly: in our protocols, every time a party receives, it already has the accumulator or the accumulator is contained within the same message. Messages with an invalid share are discarded.

\section{Subprotocols}

Our main byzantine protocol can be described as $3$ different subprotocols chained together:
\begin{itemize}
    \item All-to-Quorum Broadcast (AQB): All $n$ communicate with a quorum made of $\Theta(t)$ nodes and influence their input for the following sub-protocol (Quorum Agreement). This new part, which was missing when considering only binary agreement, allows $L$ to be decoupled from $n$.
    \item Quorum Agreement: the quorum of $\Theta(t)$ nodes runs an agreement protocol. This protocol has different properties compared to a regular Byzantine Agreement protocol, which are needed to make full use of the two other subprotocols.
    \item Quorum-to-All Broadcast (QAB): This step, which originates from the work by Constantinescu et al. \cite{constantinescu2025few} allows the output of the Quorum Agreement to be disseminated to all parties. Compared to previous work, some additional subtleties are required to send the $L$-bit value in an efficient way to parties which don't have it.
\end{itemize}

We now provide a formal description of all $3$ subprotocols.

\paragraph{All-to-Quorum Broadcast} The $n$ parties in an AQB protocol have each an input value $\vin$. There is a known set of $k = 9t+1 = \Theta(t)$ quorum parties which can output a value $v_{out}$ which can be $\ast$. An AQB protocol satisfies the following definition:
\begin{itemize}
    \item \textit{Robust termination}: All but at most $f$ honest quorum parties will eventually decide a value.
    \item \textit{Robust validity}: If all honest parties have the same input $v$, then at least $k-2t$ honest quorum parties will decide $v$.
    \item \textit{Non-amplification}: There exists $C > 0$ such that for any input $v \ne \ast$, if at least $C \cdot t$ honest parties do not have $v$ as input, then at most $t$ honest quorum parties will decide $v$.
\end{itemize}

\paragraph{Quorum Agreement}
The Quorum Agreement is a protocol run on $n$ parties. When used together with AQB and QAB we run it on the quorum made of $k = 9t+1 = \Theta(t)$ parties. Each party has an input $\vin$ and can decide an output $v_{out}$. A Quorum Agreement protocol satisfies the following properties:
\begin{itemize}
    \item \textit{Validity}: If all honest parties have the same input $v$, then only this value can be decided
    \item \textit{Agreement}: If two honest parties decide respectively $v_1$ and $v_2$, then $v_1 = v_2$.
    \item \textit{Partial Termination}: At least $n/4$ honest parties eventually decide. 
    \item \textit{Provability}: When deciding a value, parties also provide a certificate which any party can use to prove this value was decided.
    \item \textit{Representativity}: If an honest party decides a non-$\ast$ value $v$, let $n_v$ the number of honest parties with input $v$, then $n_v + f > n/3$.
\end{itemize}
In the asynchronous setting, partial termination is replaced with probabilistic partial termination where the property holds only almost surely. We remark that quorum agreement is really similar to byzantine agreement. In particular if a quorum agreement protocol also satisfies termination (or probabilistic termination in the asynchronous setting), then it is also a Byzantine Agreement protocol.

The round complexity of a Quorum Agreement protocol is the number of rounds before $n/4$ honest parties decide.

\paragraph{Quorum-to-All Broadcast}: The QAB protocol is used to broadcast the decision of the quorum of $k$ parties to all $n$ parties. More specifically, every party in the quorum can receive an input $\vin$ and a certificate $cert$. There exists a function local \textbf{check} that every party possesses such that:
\begin{itemize}
    \item If an honest quorum party receives $\vin$, $cert$, then $\textbf{check}(\vin, cert) = 1$ and it is computationally infeasible for an adversary to find $v' \ne \vin, cert'$ such that $\textbf{check}(v', cert') = 1$.
    \item If two honest quorum parties receive respectively $\vin$ and $\vin'$ as input, then $\vin = \vin'$.
\end{itemize}
Note that quorum parties may receive an input at different times or not receive an input at all. All parties can also decide an output value $v_{out}$. We say that the protocol is a QAB protocol if it satisfies the following property:
\begin{itemize}
    \item \textit{Validity}: If an honest party decides a value, it is $\vin$.
    \item \textit{Termination}: If at least $k/4$ honest quorum parties receive input $\vin$, then every honest party eventually decide $\vin$.
\end{itemize}

The round complexity of a QAB protocol is the number of rounds between the time $k/4$ parties receive an input and all honest parties decide an output. A particularity of our QAB implementation is that its bit complexity will depend on a parameter $h$ which is the number of parties that do not know $\vin$, assuming $\vin \ne \ast$. A party may already know $\vin$ (but not know it was the value decided) if it receives it as the input of a previous subprotocol (in our case the AQB). To give some foresight about it, using the AQB's non-amplification and quorum agreement's representativity, we will guarantee that $h = \mathcal{O}(t)$.

We now describe how to combine these $3$ subprotocols together:
\begin{theorem}
    Given:
    \begin{itemize}
        \item An AQB protocol with bit complexity $\textsc{BITS}_{AQB}(n,t)$ and constant round complexity.
        \item A Quorum Agreement protocol which for $k$ parties has with bit complexity $\textsc{BITS}_{QA}(k, f)$, round complexity $\textsc{ROUNDS}_{QA}(f)$ and resiliency at least $\lfloor k / 3 \rfloor$.
        \item A QAB protocol with bit complexity $\textsc{BITS}_{QAB}(n, t, f, r, h)$ where $r$ is the number of rounds for $k/4$ honest parties of the quorum to receive an input while $h$ is the number of parties which do not know the QAB input, assuming the decided value is not $\ast$, and round complexity $\textsc{ROUNDS}_{QA}(n,f)$
    \end{itemize}
    Then there exists a constant $C > 0$ and a protocol solving Byzantine Agreement with resiliency $t < n/9$, bit complexity:
    \begin{align*}
        \textsc{BITS}_{AQB}(n,t) + \textsc{BITS}_{QA}(t, 2f) + \textsc{BITS}_{QAB}(n, t, f, \textsc{ROUNDS}_{QA}(2f), C \cdot t)
    \end{align*}
    and round complexity:
    \begin{align*}
        \mathcal{O}(1) + \textsc{ROUNDS}_{QA}(2f) + \textsc{ROUNDS}_{QAB}(n, f) 
    \end{align*}
\end{theorem}

\begin{proof}
    Our Byzantine Agreement consists of using each of our subprotocol one after the other. To be more specific, the quorum is made of an arbitrary fixed set of $k = 9t + 1$ parties (for example the ones with the lowest ids). We then run the AQB protocol with the Byzantine Agreement's input, use its output as the Quorum agreement's input within the quorum and then use the Quorum's agreement output as the input of the QAB protocol. The output of the QAB protocol is the output of the Byzantine Agreement protocol.

    We first note that the conditions for the QAB protocol are satisfied: agreement from the Quorum agreement protocol guarantees that all parties will get the same input for QAB and provability shows that the certificate is correct. We now prove the different properties:
    \begin{itemize}
        \item Agreement: Guaranteed by QAB validity.
        \item Validity: If all honest parties have the same input $v$, then AQB's robust validity guarantees that at least $k - 2t \geq 4t + 1$ honest quorum parties will join quorum agreement with input $v$. Because Quorum Agreement has resiliency at least $\lfloor k / 3 \rfloor = 2t = k - (4t + 1)$, using the Quorum Agreement's validity, $v$ will be decided by the Quorum Agreement which implies that it is the only value that can be decided by the QAB and thus the Byzantine Agreement protocol because of the QAB validity property.
        \item Termination: Using AQB's robust termination, at least $k - 2f \geq k-2t$ honest quorum parties will decide a value and thus participate in the quorum agreement. We remark that $k - (k - 2t) = 2t < k/3$ which is at least the resiliency of the quorum agreement protocol. Therefore, using the Quorum Agreement partial termination, at least $k/4$ honest parties will eventually decide a value, which allows us to conclude using the QAB termination.
    \end{itemize}

    Regarding the bit/round complexity. it is the sum of the bit/round complexity of each protocol. The first thing to remark is that by definition of the quorum agreement round complexity, the number of rounds the QAB has to wait before at least $k/4$ quorum parties have an input is exactly the quorum agreement round complexity. Moreover, because of AQB's robust termination, up to $f$ honest quorum parties may not join quorum agreement (in addition to the $f$ byzantine). Therefore, $f$ must be replaced by $2f$ in quorum agreement by and round complexity.
    Finally, we need to explicit the parameter $h$ given in the theorem's definition as the number of parties which do not know the QAB input. Let $v$ be the value decided by the algorithm and assume $v \ne \ast$. Using AQB's non-amplification, there exists $C > 0$ such that if $C \cdot t$ honest parties do not have input $v$, then at most $t$ honest quorum parties get $v$ as output of the AQB. Assume by contradiction that $h > C \cdot t$, then with this non-amplification property, at most $t$ honest parties have input $v$ for the quorum agreement. We remark that $t + 2f \leq 3t \leq k/3$ so using the quorum agreement's representativity, $\ast$ will be decided, and using QAB's validity, $\ast$ will be decided by the byzantine agreement protocol. But this contradicts the fact that $v \ne \ast$ was decided by the protocol.
\end{proof}

We remark that the resulting protocol only has resiliency $t < n/9$. This non-optimal resiliency is caused by the AQB protocol. When $t \geq n/9$, we get similar results and optimal resiliency by just removing the AQB protocol from the overall algorithm (the AQB is used to decouple $n$ and $t$, which is not useful in this case).

\section{QAB in Partial Synchrony}
Let $k = \Theta(t)$ be the quorum size. Quorum Agreement guarantees that $k/4$ honest quorum nodes output a value $v^\ast$, together with a threshold-signed cryptographic accumulator of the erasure encoding of $v^\ast$. We denote this certificate by $proof$. The goal of the Quorum to All Broadcast (QAB) protocol is to disseminate the pair $\langle v^\ast, proof\rangle$ from the quorum to the rest of the system.

To keep bit complexity bounded, each quorum node $i$ locally erasure-encodes $v^\ast$ and is responsible for disseminating only the $i$-th share of the encoding.

Our main building block is a dissemination procedure called a \emph{QAB wave}. A QAB wave is parameterized by an estimate $\hat{f}$ of the number of faulty nodes. The procedure $\textsc{Wave}(\cdot,\hat{f})$ succeeds whenever $f \leq \hat{f} $, where $f$ is the (unknown) actual number of faults. Moreover, each $\hat{f}$-wave is lightweight: informally, it communicates a number of bits proportional to $\hat{f}$ and terminates within $O(\hat{f})$ rounds.

We then run a family of waves in parallel using geometrically increasing estimates:
we invoke $\textsc{Wave}(\cdot,\hat{f})$ for $\hat{f} \in \{1,2,4,\ldots,2^{\lceil \log_2 t\rceil}\}$,
wait until one invocation succeeds, and abort all remaining invocations.

Concretely, quorum node $i$ invokes $\textsc{Wave}(v,\hat{f})$ with input $v = \langle acc, s_i, proof\rangle$,
where $acc$ is the accumulator for $v^\ast$, $s_i$ is node $i$'s share of $v^\ast$,
and $proof$ is a certificate proving that $v^\ast$ was decided and its accumulator is $acc$.

\begin{dianabox}{\textsc{QABNode}}
\algoHead{QAB for Quorum Node $i$}
    \begin{algorithmic}
        \UponTrue{deciding $\langle v^\ast, proof\rangle$ in Quorum Agreement}
            \State $(acc, (s_j))\gets encode_k(v^\ast)$. 
            \State $v \gets \langle acc, s_i, proof\rangle$

            \Statex
        
            \InParallel{}
                \For{$j = 1, \ldots, \lceil \log_2 t\rceil$}
                    \State $\hat{f} \gets 2^j$
                    \State Initiate $\textsc{Wave}(v, \hat{f})$
                \EndFor
            \EndParallel
        \EndUpon

        \UponTrue{some wave is done}
            \State \textbf{Terminate}
        \EndUpon
     \end{algorithmic}
\end{dianabox}

Let us now describe what comprises a QAB Wave.

\subsection{QAB Wave}
For an $\hat{f}$-wave, we restrict the communication between nodes to a specific graph, which we call a \emph{communication graph}. If there is an edge between two nodes in this graph, we say that these two nodes are \emph{linked}. 

We construct the communication graph $G$ as follows. First, we assign parties to committees according to Theorem \ref{thm:nodes2committees} with parameters $n$ and $\hat{f}$. This way, we obtain a set $\mathcal{C}$ of $O(\hat{f}\log n)$ committees with each party being assigned to $O(\log n)$ committees. For each committee $C \in \mathcal{C}$, we choose (arbitrarily) a node in $C$ ``responsible'' for this committee and call such a node a \emph{relayer} for $C$, denoted $relayer(C)$. Hence, we have $O(\hat{f}\log n)$ relayers. In the communication graph, we then link every party to the relayer of each committee this party is assigned to. That is $\{(u, relayer(C))\mid u \in C, C \in \mathcal{C}\}  \subseteq E(G)$. We also link every relayer to every quorum node.

\begin{figure}[ht]
    \centering
    \includegraphics[width=0.5\linewidth]{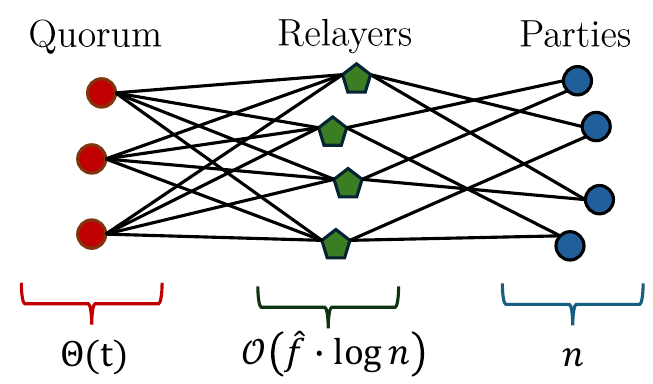}
    \caption{The communication of an $\hat{f}$-wave. Nodes are assigned three roles (non-exclusive):  \textit{quorum}, \textit{relayers} and \textit{parties},  with respective sizes of $3t + 1$, $O(\hat{f}\log n)$ and $n$. Each party is only linked to $O(\log n)$ relayers, and each relayer is linked to each quorum node.}
    \label{fig:communication graph}
\end{figure}

Given the communication graph, the high-level description of our QAB wave with parameter $\hat{f}$ is as follows.

\begin{enumerate}  
    
    \item \textbf{Quorum-Relayers: Disseminate.} Each quorum node, upon invoking $\textsc{Wave}(v, \hat{f})$, starts disseminating $v$ to relayers. This is done by sending $v$ to subgroups of relayers of size $O(\log n)$ in intervals of $\Delta$ time.

    \item \textbf{Quorum-Relayers: Ping.} After sending $v$ to all the relayers, a quorum node starts to ping them. This is done by sending "ping" to subgroups of relayers of size $O(\log n)$ in intervals of $\Delta$ time.

    \item \textbf{Relayers-Parties.} Upon receiving and verifying a share from a quorum node, the relayer stores it. Once it has $K/4$ distinct shares, it reconstructs the value $v^\ast$ decided in the Quorum Agreement. Then it sends $v^\ast$ with a proof $proof$ to the parties in its committee.

    \item \textbf{Parties-Relayers-Quorum.} Every committee has an aggregate signature setup. After receiving $(v^\ast, proof)$ from a relayer and verifying it, each party produces a partial signature for a message saying that it knows $v^\ast$ and sends it back to the relayer. After receiving a valid partial signature from every party linked to them, a relayer aggregates the signatures into a single one denoted $\Sigma$. Once a relayer receives a "ping" message from a quorum node, it responds with $\Sigma$.

    \item \textbf{Quorum-Undecided.} Each quorum node $u$ maintains a set $\mathrm{Acknowledged}$, initially empty, of parties of whom $u$ is aware that they know $v^\ast$. Once $u$ receives and verifies an aggregated signature for a committee $C$ from a $relayer(C)$, $u$ knows that all the parties in $C$ know $v^\ast$, and hence it updates $\mathrm{Acknowledged} = \mathrm{Acknowledged} \cup C$. Once $|\mathrm{Acknowledged}| \geq |n - c_b\cdot \hat{f}|$ where $c_b$ is a constant from Theorem \ref{thm:nodes2committees}, $u$ sends $v$ directly to the parties in $[n] \setminus \mathrm{Acknowledged}$ one by one in intervals of $\Delta$ time and concludes the wave done.
\end{enumerate}

As a remark, note that in the last step, when $u$ is sending $v$ to the parties in $[n] \setminus \mathrm{Acknowledged}$, this happens \emph{not} according to the communication graph. For $[n] \setminus \mathrm{Acknowledged}$ is determined dynamically. This is the only such place; all other messages are sent according to the graph.

The pseudocode and analysis for our QAB protocols are provided in \cref{section:QAB-appendix}.

\section{Moving from \texorpdfstring{$Ln$}{Ln} to \texorpdfstring{$Lt$}{Lt}}

The main idea for improving the $O(Ln)$ term to $O(Lt)$ in the bit complexity of our consensus algorithm is the following.
We ensure that quorum nodes propose a value to the Quorum Agreement (QA) protocol only if that value is highly prevalent in the system, i.e., held by most parties; otherwise, they propose $\ast$.
Consequently, if QA decides a non-$\ast$ value $v$, we can show that $v$ is possessed by most parties.
In this case, relayers need not transmit $v$ to those parties that already hold it, yielding bit complexity savings.
On the other hand, if QA decides $\ast$, this indicates that parties hold many different values; thus, $\ast$ is admissible by Strong Unanimity.
Moreover, broadcasting $\ast$ from the quorum to all parties is inexpensive.

The main technical ingredient of this section is an \emph{All-to-Quorum Broadcast (AQB)} procedure, which enables quorum nodes to determine whether there exists a prominent value in the system, or whether no such value exists.

\subsection{All to Quorum Broadcast}
For the AQB, we use a communication graph that resembles that of a QAB wave, however, we would require more from the committee assignment.  To this end, we use a disperser-based committee assignment similar to Theorem \ref{thm:nodes2committees} to obtain an additional property. We say that a committee is covered by a set of parties if at least one of its parties is in this set.
\begin{restatable}{theorem}{nodesToCommitteesPlusThm}
    \label{thm:nodes2committeesPlus}
    There exist $C > 0, D \geq 1, F > 9D$ such that, given $n$ parties and $t \leq n$, there is a way to assign parties to committees such that:
    \begin{enumerate}
        \item There are $4F \cdot (9t + 1) \cdot \log n$ committees.
        \item Every party is assigned to $D \cdot \log n$ committees.
        \item Any subset of parties $S \subset [n]$ of size at least $C \cdot t$ covers at least $90\%$ of all committees.
    \end{enumerate}
\end{restatable}

The proof is deferred to the Appendix \ref{sec:appendix:proofs}. The communication graph of AQB is then defined as follows.

\paragraph*{Communication Graph of AQB}
First, we assign parties to committees according to Theorem \ref{thm:nodes2committeesPlus} with parameters $n$ and $t$. This way, we obtain a set $\mathcal{C}$ of $O(t\log n)$ committees with each party being assigned to $O(\log n)$ committees. For each committee $C \in \mathcal{C}$, we choose (arbitrarily) a node in $C$ ``responsible'' for this committee and call such a node a \emph{relayer} for $C$, denoted $relayer(C)$. Hence, we have $R := O(t\log n)$ relayers. In the communication graph, we then link every party to the relayer of each committee this party is assigned to. That is $\{(u, relayer(C))\mid u \in C, C \in \mathcal{C}\}$ is a subset of edges. Now, denote a quorum size to be $k = 9t+1$. We split relayers into $k$ batches, each of size $R/k = 4F \log n$. We then link every relayer of the $i$-th batch to the $i$-th quorum node.

\begin{figure}[ht]
    \centering
    \includegraphics[width=0.5\linewidth]{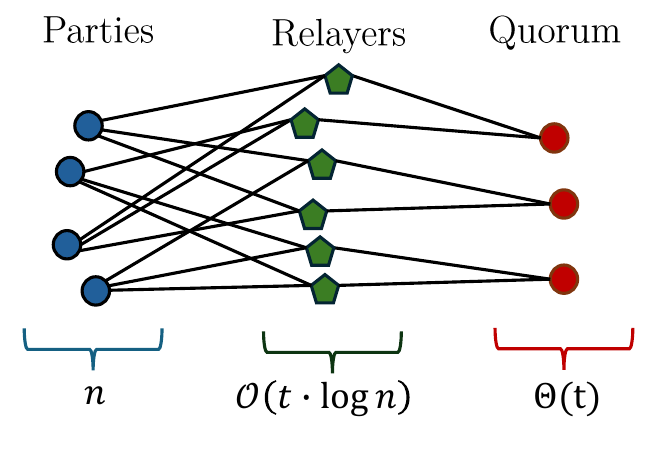}
    \caption{The communication of an AQB. Nodes are assigned three roles (non-exclusive):  \textit{quorum}, \textit{relayers} and \textit{parties},  with respective sizes of $k := O(t)$, $R := O(t\log n)$ and $n$. Each party is only linked to $O(\log n)$ relayers. Relayers are split into $k$ roughly equal batches of size $R/k = \Theta(\log n)$. All relayers from the $i$-th batch are linked to the $i$-th quorum node.}
    \label{fig:communication graph aqb}
\end{figure} 

We now describe the AQB protocol that uses this graph.

\paragraph*{AQB Protocol}
The protocol is two-hop:
\begin{enumerate}
    \item \textbf{Parties-Relayers:} Each party sends a hash of its input value to its relayer
    \item \textbf{Relayers-Quorum:} After receiving an hash from all its committee member, if the relayer always received the same hash, it forwards it to its designated target quorum party. Otherwise it sends $\ast$ instead.
    \item \textbf{Quorum Propose:} Upon receiving the same hash $h$ from $50\%$ of the relayers in its batch, and having as input a value $v$ with $hash(v) = h$, decide $v$ (and use input as input to the quorum agreement). If the quorum node received already messages from $75\%$ of its relayers and the previous event did not happen yet, then decide $\ast$ instead.
\end{enumerate}

We show in \cref{section:AQB-appendix} that this is an AQB protocol which supports an arbitrary number of byzantine parties, works deterministically in asynchrony and has constant latency (last $2$ rounds in synchrony).

\section{Quorum agreement}
\subsection{Synchrony}

Our quorum agreement protocol builds on the ones described by \cite{constantinescu2025few}, but with multiple improvements to handle the fact that the value decided is an $L$-bit value instead of a binary one.

\begin{itemize}
    \item \textbf{(Dis-)Agreeing on the hash}: If at least $t+1$ honest parties have the same hash as input, a threshold signature for it can be obtained, which proves that that value associated with this hash can be decided. If this is not the case, we need to get what we call an \textit{evidence of disagreement} which shows that not all honest parties have the same input and therefore any value can be decided.
    \item \textbf{Evidence of disagreement in Partial Synchrony}: In partial synchrony, because $t < n/3$, we can reuse the Big-Bucket approach from \cite{rambaud2022bucket}. It can guarantee a threshold signature for a value or an evidence of disagreement will be obtained as long we are after GST and the leader is honest.
    \item \textbf{Evidence of disagreement in Synchrony}: In synchrony, because of the higher resiliency $t < n/2$, the previous Big-Bucket approach alone may fail to return a valid proof. Instead, we combine the Big-Bucket approach with the idea from \cite{constantinescu2025few}, where parties drop their input if they realize there is a disagreement among honest parties. To be more specific, we use the latter idea with hashes which appear more than $n/4$ times and the Big-Bucket approach with other values. We show that if $f < n/4$, then after $f+1$ views with honest leader, one of the leader will be able to retrieve either a proof of agreement of a proof of disagreement.
    \item \textbf{Distributing shares instead of the whole value}: If we were to reuse the approach from \cite{constantinescu2025few} while replacing the binary decision by the whole $L$-bit value, the resulting bit complexity would be $\mathcal{O}(n \cdot f \cdot L)$. Instead, we use error-correction and cut the input value into $n$ parts such that the original value can be recovered using any subset of size at least $n/4$. Carefully sharing only some of these shares instead of the full value allows us to reduce the bit complexity to $\mathcal{O}(n\cdot (f \cdot \kappa + L))$.
    \item \textbf{Improved Complaining}: We note that the complaining approach from \cite{constantinescu2025few}, which was used to achieve $\mathcal{O}(f)$ round complexity, does not work with $L$-bit inputs as it would incur a $\mathcal{O}(n \cdot f \cdot L)$ bit complexity. Instead, we use the expander approach from the DARE to Agree paper \cite{civit2024dare}. To be more specific, once a party gets the full value, it starts spreading it in the expander network. Meanwhile, if a party receives a proof for the hash, it waits for $\mathcal{O}(\log n)$ rounds, then asks every party for their shares. This approach guarantees that if $f = o(n)$, then every honest party decides within $\mathcal{O}(f + \log n)$ rounds.
\end{itemize}

\subsection{Retrieval protocols}

In their work \cite{constantinescu2025few}, Constantinescu and al. abstracted away getting a value that satisfied strong unanimity using what they called a retrieval protocol. To be more specific, during a view, a leader could run a retrieval protocol to obtain a value and proof that satisfies some predicate, similar to external validity. A key idea is that this retrieval protocol will only have to be run $\mathcal{O}(f)$ times, meaning it is acceptable for it to have $\mathcal{O}(n)$ message complexity for binary agreement.

\begin{definition}[Adapted from (Definition 7.1)\cite{constantinescu2025few}]
    A retrieval protocol consists of two subprotocols \textsc{RetrieveLeader} and \textsc{RetrieveParty}($\vin$). \textsc{RetrieveLeader} is called by a leader while \textsc{RetrieveParty} is called by every honest party. We assume parties know who the leader is and if the leader is honest, then all honest parties call \textsc{RetrieveParty} at most one round after the leader called \textsc{RetrieveLeader}. A call to \textsc{RetrieveLeader} along with its associated \textsc{RetrieveParty} is called an instance. We say that $(\textsc{RetrieveLeader},\textsc{RetrieveParty}(\vin))$ is a retrieval procedure if it satisfies the following properties
    \begin{itemize}
        \item \textsc{RetrieveLeader} either returns $\bot$ or a $proof$.
        \item If called after GST by $f+1$ different honest leaders, this protocol will at least once \textbf{not} return $\bot$.
    \end{itemize}
\end{definition}

\subsubsection{Partially Synchronous retrieval protocol} \label{sec:gst-retrieval}

In the partially synchronous setting, we can directly use the big bucket approach from Rambaud et al. \cite{rambaud2022bucket}. To be more specific, if we cannot get a $(t+1)$-threshold signature on an accumulator, we partition the accumulator range $[\![ 0, 2^\kappa - 1 ]\!]$ into $\mathcal{O}(1)$ intervals such that the input of at least $t+1$ honest party is outside any given interval.

\begin{lemma}\label{lemma:input-inter-partition}
    Given a constant $X >= 2$ and a multiset $S$ of at most $n$ values in $[\![ 0, 2^\kappa - 1 ]\!]$ such that every value appears at most $n/X$ times, it is possible to partition the input range into a set $\mathcal{I}$ of $\mathcal{O}(1)$ intervals such that each interval contains at most $n/X$ elements in $S$.
\end{lemma}

\begin{proof}
    We partition using a greedy approach. To do so, we consider an interval starting at $0$ and give its the maximal range possible while containing at most $n/X$ elements from $S$. We then add the resulting interval $I = [l,r]$ to $\mathcal{I}$ then repeat this approach again starting from $r+1$ until the whole input range has been partitioned.

    By construction, each interval contains at most $n/X$ elements from $S$. Moreover, when considering two consecutive intervals in $\mathcal{I}$, we claim that one of them contains at least $n/(2X)$ elements from $S$. Indeed, if that were not the case, then the greedy approach would have merged these two intervals. Because there can be at most $2X$ intervals containing at least $n/(2X)$ elements, we conclude that $|\mathcal{I}| \leq 4X + 1 = \mathcal{O}(1)$.
\end{proof}

We now provide of description of the retrieval protocol for the partially synchronous setting. Basically, if the leader is able to get a $(t+1)$-threshold signature for a given accumulator (an evidence of agreement), it then reconstructs the value associated with the accumulator from the shares it received. In the last step, it broadcasts the full value to every party and gets a $n-t$ threshold signature for parties acknowledging it received the value.
If it could not get an evidence of agreement, it divides the input range using \cref{lemma:input-inter-partition} and asks parties to provide a threshold signature for intervals which do \textbf{not} contain their input. We will show that assuming the protocol is run after GST, this will always succeed.

Pseudocode and analysis of our partially synchronous retrieval protocol are provided in \cref{section:gst-retrieval-appendix}.

\subsubsection{Fully synchronous retrieval protocol}

In the fully synchronous setting, because of the higher resilience $t < n/2$, the Big-Bucket approach on its own will fail. For example, if $n/3$ parties have input $0$ and $n/3$ parties have input $1$ while the remaining one crash, it will be impossible to partition the input range into intervals such that for each interval, there are $t+1$ parties outside this interval.

To work around this issue, we combine big buckets with the approach from \cite{constantinescu2025few}. To explain this latter approach, if the leader is not able to get a threshold signature for an accumulator, it means honest parties do not share the same value (even though the leader is not able to prove it). In this case the leader drops its input (he sets it to $\bot$). The idea is that these small changes (dropping your input) accumulate and allow for a proof to be made within $\mathcal{O}(f)$ rounds. The issue is that the approach for \cite{constantinescu2025few} was given for binary agreement and in general only works for binary decisions. By coupling these two previous approaches, we obtain a retrieval protocol for the synchronous setting. To be more specific, after parties send their hash to the leader, it does the following:
\begin{itemize}
    \item For any accumulator $acc$ that appears at least $n/4$ times, it uses the idea from \cite{constantinescu2025few} separately on each hash.
    \item For the remaining accumulators, which appear less than $n/4$ times, we use the big bucket approach on them.
\end{itemize}

We now describe our proof of disagreement. 
\begin{definition}
    For the synchronous setting, an evidence of disagreement consists of a set $A \subset 2^\kappa$ such that $|A| \leq 4$, a partition of the hash range into intervals $\mathcal{I}$ such that $|\mathcal{I}| \leq 9$ and $(t+1)$-threshold signatures for every element such that:
    \begin{itemize}
        \item For $a \in A$, an honest party provide a partial signature for $a$ if its input's accumulator is not $a$.
        \item For $I \in \mathcal{I}$, an honest party provides a partial signature for $I$ if its input's accumulator is not in $I\setminus S$.
    \end{itemize}
\end{definition}
We remark that if an honest party sets its input to $\bot$, it will provide partial signatures for any element requested.

Pseudocode and Analysis of our synchronous retrieval protocol are given in \cref{section:sync-retrieval-appendix}.

\subsection{Partially Synchronous Algorithm}

We can now describe our algorithm achieving partial termination in the partially synchronous setting. To do so, we use the view-based byzantine agreement from Constantinescu et al. \cite[algorithm \textsc{ViewByzantineAgreement}]{constantinescu2025few} and replace their retrieval protocol with our own defined in \cref{sec:gst-retrieval}. At the end, if the party has a disagreement evidence, it decides $\ast$. If it has an agreement proof for an accumulator $acc$ and knows the value associated with $acc$ (it received it from an $INFORM$ message in \textsc{RetrievalPartyGST}), it decides it. Otherwise it does not decide anything.

\begin{dianabox}{\textsc{QuorumGST}}
\algoHead{Adaptive protocol for partial agreement with a party $p$ and input $\vin$}
\begin{algorithmic}
    \State $(evidence, proof) \gets \textsc{ViewByzantineAgreement}(\vin)$
    \If{$evidence$ is a disagreement evidence}
        \State Decide $(\ast, proof)$
    \Else
        \State $evidence$ is an agreement evidence for an accumulator $acc$
        \If{$p$ received an input $v$ with accumulator $acc$ from an $INFORM$ message}
            \State Decide $(v, proof)$
        \Else
            \State Do not decide anything (The Quorum agreement never returns)
        \EndIf
    \EndIf
\end{algorithmic}
\end{dianabox}

In \cref{section:gst-quorum-appendix}, we prove that \textsc{QuorumGST} achieves agreement, strong unanimity, Provability and partial termination in $\mathcal{O}(f)$ rounds with $\mathcal{O}(n \cdot (L + f \cdot \kappa))$ bit complexity.

\subsection{Synchronous Algorithm}

Due to the higher resiliency $t < n/2$, our quorum protocol for the synchronous setting is a lot more complex compared to the partially synchronous setting. In particular, we need to modify the \textsc{ViewByzantineAgreement} protocol from \cite{constantinescu2025few} to ensure an optimal bit complexity. We now give the main ideas and main points describing how we modified our protocol to accommodate for $L$-bit inputs:

\textbf{Use of the retrieval protocol}
If the retrieval protocol returns a non-$\bot$ value to an honest leader, the leader may still not be able to get a commit proof if $f$ is close to $t$ and the byzantine parties do not cooperate. Nevertheless, it will still send its proof to the subsequent leaders until a commit proof is created. This ensures that at most one honest leader will have its retrieval protocol return a non-$\bot$ value and is necessary to bound the bit complexity.

\textbf{Complaining is not enough}
In the original paper, a byzantine leader could get a commit proof and only distribute it to the first half of the honest parties. To work around this issue, Constantinescu et al. add a complain mechanism where parties can complain and get the value back if the leader is honest and already has the commit proof. In our case, if the leader were to send the full value to parties which complain, it would incur a $\mathcal{O}(n \cdot f \cdot L)$ bit complexity. If the leader were to send only a share, even with $f = 1$ a party might take $\Omega(n)$ views to reconstruct the decided value. We therefore use a different approach to share the value using an \emph{expander}.

\textbf{Expander-based dispersion}
We use the same approach as in the paper DARE to Agree \cite{civit2024dare}: once receiving a value as well as a commit proof for it, a party will share it using an expander. The expander used has constant degree and ensures that if $f < \alpha \cdot n$ and at least $\beta \cdot n$ honest parties start sharing the value, with $\alpha, \beta > 0$ constants, then every but $\mathcal{O}(f)$ parties will learn it within $\mathcal{O}(\log n)$ rounds.

\textbf{Dispersal proof}
We note that even with the help of an expander, $\mathcal{O}(f)$ honest parties may not get the value. In their paper, Civit et Al. make parties wait for $\mathcal{O}(n + \log n)$ rounds and if they still did not get the value, they ask for it. We note that this approach does not work in an adaptive setting, as parties would have to wait $\mathcal{O}(f + \log n)$ rounds, but they do not know $f$. Instead, on top of the existing proofs created in the view protocol (key then lock then commit), we add a fourth proof called a dispersion proof. The main idea is that if a party receives a dispersion proof, then assuming $f < n/8$, it knows than $\Omega(n)$ parties have already started dispersing their value and thus after an additional $\mathcal{O}(\log n)$ rounds, only $\mathcal{O}(f)$ parties won't have the value.

\textbf{Improved fallback}
As in the protocol for the binary variant, if $f > n/4$, then the views based agreement protocol may not achieve termination. We have to detect when this happens and run a fallback protocol when this happens. There are two differences: first we must use Nayak et al. \cite{nayak2020extension} protocol for MVBA as a fallback. Second, we must be careful about how parties distribute shares of the decided value to ensure the bit complexity does not exceed $\mathcal{O}(n \cdot (L + f \cdot \kappa))$.

As stated above, our protocol relies on expander graph, more precisely the one described by Upfal \cite{upfal1992expander}:

\begin{theorem}[\cite{upfal1992expander}]\label{theo:sync-expander}
    There exists $D \geq 1, c \geq 1$ such that for any $n$, there exists a $D$-regular graph $G$ on $n$ vertices, such that for any subset of vertices $T$ with $|T| \leq n/72$, $G \setminus T$ has a connected component of size at least $n - 6 |T|$ and diameter at most $c \cdot \log n$. Moreover, $G$ can be computed deterministically in polynomial time.
\end{theorem}

We remark that the construction of $G$ in \cite{upfal1992expander} is only given for $n = p + 1$, where $p$ is a prime number. Using a more recent result $\cite{cohen2016ramanujan}$ allows such a graph to be built for any $n$.

In this protocol, we use a graph expander $G$ on the $n$ parties with constant $D, c \geq 1$. 

In \cref{section:sync-quorum-appendix}, we provide the pseudocode for our protocol and prove that it satisfies agreement, strong unanimity, Provability, partial termination in $\mathcal{O}(f)$ rounds, termination in $\mathcal{O}(f + \log n)$ rounds and has $\mathcal{O}(n \cdot (L + f\cdot  \kappa))$ bit complexity.

\subsection{Asynchrony} \label{subsec:AsyncQAMain}

In asynchrony, as was proven by Constantinescu et al. \cite{constantinescu2025few}, byzantine algorithms cannot be adaptive in their communication complexity. As such, existing algorithms are already near-optimal when $t = \Theta(n)$. Our main contribution is to extend these results to $t = o(n)$, which requires a quorum agreement protocol. We explain how to obtain a quorum agreement protocol from an existing MVBA protocol. The protocol is described in \cref{app:asyncQA}.

\begin{theorem}
    There exists a Quorum Agreement protocol in asynchrony with resiliency $t < n/3$, expected constant latency and expected bit complexity $\mathcal{O}(n \cdot L + n^2 \cdot \kappa)$.
\end{theorem}

\section{Conclusion}
We analyze how much communication and how many rounds are required for deterministic multi-value Byzantine agreement under the standard spectrum of timing assumptions: fully synchronous, partially synchronous, and fully asynchronous settings.

For the large-network setting where the number of participants is much larger than the maximum number of Byzantine faults, we develop adaptive protocols that enjoy almost optimal message complexity while matching the asymptotically optimal round complexity.

A main technical tool is a deterministic method for selecting committees built from bipartite disperser constructions. This yields a scalable way to spread information efficiently while staying resilient to adaptive corruption strategies.

Our protocols for the partially synchronous model currently depend on ideal clock alignment. A natural next step is to clarify how varying degrees of clock synchronization influence what properties can be achieved. In particular, can adaptive agreement be obtained without assuming any bound on clock skew?

\bibliographystyle{plainurl}

\appendix

\section{Quorum to All Broadcast}\label{section:QAB-appendix}

\subsection{Pseudocode}

\begin{dianabox}{\textsc{QuorumNodeWave}}
\algoHead{$\hat{f}$-Wave for Quorum Node}
    \begin{algorithmic}
        \UponTrue{Initiated $\textsc{Wave}(v, \hat{f})$}
            \State $R_{\hat f} \gets$ relayers of $\hat{f}$-wave 
            \State $\mathrm{relayer\_batches}\gets$ split $R_{\hat f}$  into $\hat f$ batches of size $O(\log n)$
            \For{$\mathrm{relayer\_batch} \in \mathrm{relayer\_batches}$}
                \State Send $v$ to relayers in $\mathrm{relayer\_batch}$
                \State Wait $\Delta$
            \EndFor
            \State

            \While{true}
                \For{$\mathrm{relayer\_batch} \in \mathrm{relayer\_batches}$}
                    \State $\tau \gets$ current time stamp
                    \State Send $\langle \mathrm{''ping''}, \tau\rangle$  to relayers in $\mathrm{relayer\_batch}$
                    \State Wait $\Delta$
            \EndFor
            \EndWhile
        \EndUpon

        \Statex

        \UponTrue{receive valid $\Sigma$ from relayer of committee $C$}
            \State $\mathrm{Acknowledged} \gets \mathrm{Acknowledged} \cup C$ 
        \EndUpon

        \Statex

        \UponTrue{$|\mathrm{Acknowledged}| \geq n - c_b\cdot \hat{f}$}
            \For{$p \in [n] \setminus \mathrm{Acknowledged}$}
                \State Send $v$ to $p$
                \State Wait $\Delta$
            \EndFor
            \State \textbf{Wave done}
        \EndUpon
    \end{algorithmic}
\end{dianabox}

\begin{dianabox}{\textsc{RelayerWave}}
    \algoHead{$\hat{f}$-Wave for Relayer of Committee $C$}
    \begin{algorithmic}[1]
        \UponTrue{receive valid $\langle acc, \mathrm{share}, proof \rangle$ from quorum node $i$}
            \State Store $proof$
            \State $\mathrm{shares} \gets \mathrm{shares} \cup \{\mathrm{share}\}$
        \EndUpon

        \Statex

        \UponTrue{$|\mathrm{shares}| \geq K/4$}
            \State $v^\ast \gets decode(\mathrm{shares})$
            \State Send $\langle v^\ast, proof \rangle$ to parties in $C$
        \EndUpon

        \Statex

        \UponTrue{receive valid $\sigma$ from a party in $C$}
            \State $\mathrm{sigs} \gets \mathrm{sigs} \cup \{\sigma\}$
        \EndUpon

        \Statex

        \UponTrue{Received a valid signature from every party in $C$}
            \State $\Sigma \gets$ aggregate signature of $\mathrm{sigs}$
        \EndUpon

        \Statex

        \UponSimple{$\Sigma \neq \bot$ and receive $\langle \mathrm{''ping''}, \tau\rangle$ from quorum node $q$}
            \If{$\tau \geq \text{current time stamp} - \Delta$  \textbf{and} didn't send $\Sigma$ to $q$} \label{alg line:check time stamp}
                \State Send $\Sigma$ to $q$
            \EndIf
        \EndUpon
    \end{algorithmic}
\end{dianabox}

\begin{dianabox}{\textsc{PartyWave}}
    \algoHead{$\hat{f}$-Wave for Party $p$}
    \begin{algorithmic}[1]
        \UponTrue{receive valid $\langle v^\ast, proof \rangle$ from relayer $r$ of committee $C$ s.t. $p \in C$}
            \State $\sigma \gets$ sign a share for message "I know $v^\ast$" with Aggregate Signature Scheme of $C$
            \State Send $\sigma$ to $r$
            \State \textbf{Decide} $v^\ast$
        \EndUpon

        \Statex

        \UponTrue{receive a valid $\langle acc, \mathrm{share}, proof \rangle$ from a quorum node $q$}
            \State $\mathrm{shares} \gets \mathrm{shares} \cup \{\mathrm{share}
            \}$
        \EndUpon

        \Statex

        \UponTrue{$|\mathrm{shares}| \geq K/4$}
            \State $v^\ast \gets decode(\mathrm{shares})$
            \State \textbf{Decide} $v^\ast$
        \EndUpon
    \end{algorithmic}
\end{dianabox}

\subsection{Analysis}

\begin{lemma}
    \label{lem:relayers ready}
    Let $f$ be the actual number of faults in a given run, and let $\hat{f}$ be the minimal power of two such that $\hat{f} \geq f$. Then every non-blocked relayer of $\hat{f}$-wave will create an aggregate signature of its committee within $O(f)$ rounds after GST.
\end{lemma}
\begin{proof}
    Quorum Agreement (QA) guarantees that at least $k/4$ honest quorum nodes will decide a value $v^\ast$ in QA within $O(f)$ rounds after GST. After a quorum node decides a value in QA, it takes it $\hat{f} = O(f)$ rounds to send its share to all the relayers of the $\hat{f}$-wave. Therefore, every relayer of an $\hat{f}$-wave will receive $k/4$ distinct shares within $O(f)$ rounds after GST. And will reconstruct $v^\ast$ according to Corollary \ref{cor:reed-solomon}. Consequently, any non-blocked relayer of $\hat{f}$-wave will get an aggregate signature of its committee within $O(f)$ rounds after GST.
\end{proof}

\begin{lemma}
    \label{lem:quorum terminate}
    Let $f$ be the actual number of faults in a given run. Let $q$ be an honest quorum node that decides in Quorum Agreement in round $R$. Then, $q$ terminates in round $\max\{\mathrm{GST}, R\} + O(f)$. 
\end{lemma}
\begin{proof}
    Let $\hat{f}$ be the minimal power of two such that $\hat{f} \geq f$. According to Lemma \ref{lem:relayers ready}, all non-blocked relayers of $\hat{f}$-wave will have their aggregate signature by the round $\mathrm{GST} + O(f)$, and hence by the round $\max\{\mathrm{GST}, R\} + O(f)$. Further, $q$ will start pinging all the relayers of $\hat{f}$-wave by the round $R + \hat{f} = R + O(f)$ and it takes $\hat{f} = O(f)$ rounds to do a full round of pinging, hence it will receive all the signatures from non-blocked relayers by the round $\max\{\mathrm{GST}, R\} + O(f)$.

    According to Theorem \ref{thm:nodes2committees}, at most $c_b\cdot \hat{f}$ parties are blocked in the $\hat{f}$-wave, therefore, $q$ will receive an acknowledgment from at least $n - c_b\cdot\hat{f}$ parties, hence will send its value to at most $c_b\cdot\hat{f} = O(f)$ parties in $O(f)$ additional rounds. Then it will call the $\hat{f}$-wave done, and thus will terminate.
\end{proof}

We can now prove that all the parties will decide a value fast. 

\begin{lemma}
    Let $f$ be the actual number of faults in a given run. Assume $v^\ast$ is decided in Quorum Agreement. Then all honest parties will decide $v^\ast$ within $O(f)$ rounds after GST.
\end{lemma}
\begin{proof}
    Quorum Agreement (QA) guarantees that $K/4$ honest quorum nodes, denote them $Q'$, will decide in QA within $O(f)$ rounds after GST. Lemma \ref{lem:quorum terminate} then tells us that $K$ honest quorum nodes terminate within $O(f)$ rounds after GST.

    Let $ack_q$ be the Acknowledged set of an honest quorum node $q$ at a time it terminates. We are guaranteed that every party in $ack_q$ knows $v^\ast$, and hence all the nodes in $\bigcup\limits_{q \in Q'}ack_q$ know $v^\ast$. What is left to show is that nodes in $[n]\setminus \bigcup\limits_{q \in Q'}ack_q$ also decide. Recall that before terminating, $q$ sends its share to parties in $[n] \setminus ack_q$. Therefore, parties in $\bigcap\limits_{q\in Q'}[n] \setminus ack_q$ receive $K$ distinct shares and therefore can restore $v^\ast$ and decide. But 
    \begin{align*}
        \left[\bigcap\limits_{q\in Q'}[n] \setminus ack_q\right] \cup\left[\bigcup\limits_{q \in Q'}ack_q\right] =[n],
    \end{align*}
    hence, every honest party decides within $O(f)$ rounds after GST.

    Finally, every party will decide $v^\ast$ and not any other value since before deciding, parties perform a check (either a certificate for a full value from a relayer, or certificates for a share from a quorum node), and only $v^\ast$ can pass these checks.
\end{proof}

We now move on to analyzing the bit complexity of the protocol. We do so by analyzing a bit complexity of communication between nodes of different roles of our protocol. 

\begin{lemma}
    \label{lem:bit complexity quorum relayers}
    In total, honest quorum nodes send at most $O(L\cdot f\cdot\log t\cdot \log n + f\cdot\log t\cdot \log n\cdot \kappa)$ bits to relayers after GST.
\end{lemma}
\begin{proof}
    From Lemma \ref{lem:quorum terminate}, we know that a quorum node operates in QAB for at most $O(f)$ rounds after GST. In each of these rounds, in each wave, it sends its value to $O(\log n)$ relayers. For $O(t)$ quorum nodes, this yields $O(t\cdot f\cdot\log t\cdot \log n \cdot |v|)$ bits where v is the size of a value.
    
    The value comprises a share which is $O(L/k + \kappa) = \mathcal{O}(L/t + \kappa)$ bits, an accumulator and a proof, which are $\mathcal{O}(\kappa)$ bits. So the value is of size $O(\frac{L}{t} + \kappa)$ bits. Substituting into the above, we get the expected bit complexity.
\end{proof}

\begin{lemma}
    \label{lem:bit complexity relayers parties}
    In total, honest relayers send at most $O((L + \log t \cdot \kappa)\cdot n \cdot \log n \cdot \log t)$ bits to parties after GST. 
\end{lemma}
\begin{proof}
    Independent of a wave number, one party belongs to $O(\log n)$ committees, making it $O(n \log n)$ links between parties and relayers in the communication graph for each wave. In each wave, along each link, a relayer sends a decided value and a certificate only once.
    
    Adding that this happens for $O(\log t)$ waves, we evaluate the total bits sent from relayers to parties to be $O((L + \log t \cdot \kappa)\cdot n \cdot \log n \cdot \log t)$.
\end{proof}

\begin{lemma}
    \label{lem:bit complexity parties relayers}
    In total, honest parties send at most $O(n \cdot \log n \cdot \log t \cdot \kappa)$ bits to relayers after GST. 
\end{lemma}
\begin{proof}
     Independent of a wave number, one party belongs to $O(\log n)$ committees, making it $O(n \log n)$ links between parties and relayers in the communication graph for each wave. In each wave, along each link, a sends its signature. Adding that this happens for $O(\log t)$ waves, we evaluate the total bits sent from parties to relayers to be $O( \kappa\cdot n \cdot \log n \cdot \log t)$.
\end{proof}

\begin{lemma}
    \label{lem:bit complexity relayers quorum}
    In total, honest relayers send at most $O(t\cdot f\cdot \log n \cdot \log t \cdot \kappa)$ bits to quorum nodes after GST.
\end{lemma}
\begin{proof}
    Here we distinguish between two types of messages from relayers: (I) to honest quorum nodes and (II) to malicious quorum nodes. We start by analyzing messages of type (I).

    From Lemma \ref{lem:quorum terminate}, we know that a quorum node operates in QAB for at most $O(f)$ rounds after GST. In one wave, per round, it pings $O(\log n)$ relayers. This means that a single honest quorum node pings in total at most $O(f\cdot \log n \cdot \log t)$ relayers in the time period $[\mathrm{GST}-\Delta, +\infty)$, and so the total number of pings from honest nodes in the time period $[\mathrm{GST}-\Delta, +\infty)$ is $O(t\cdot f\cdot \log n \cdot \log t)$. Since a relayer only responds to a ping sent at most $\Delta$ time before, after GST relayers only respond to pings from the time period $[\mathrm{GST}-\Delta, +\infty)$. A relayer responds by sending a signature of size $\kappa$, thus we conclude that the total number of bits of type (I) sent after GST is $O(t\cdot f\cdot \log n \cdot \log t \cdot \kappa)$.

    Now, for messages of type (II), notice that a relayer only responds to a ping at most once. Since there are $O(t \cdot \log n)$ relayers in total and $f$ malicious nodes, we conclude that the bit complexity of messages of type (II) is at most $O(t\cdot f \cdot \log n \cdot \kappa)$.
\end{proof}

\begin{lemma}
    \label{lem:bit complexity quorum parties}
    In total, honest quorum nodes send at most $O(L\cdot f\cdot \log t + t\cdot f \cdot \log t \cdot \kappa)$ bits to parties after GST.
\end{lemma}
\begin{proof}
    By Lemma \ref{lem:quorum terminate}, each honest quorum node only operates for $O(f)$ rounds after deciding in Quorum Agreement. In one round, it sends its value to at most one party in a given wave, hence to at most $O(\log t)$ parties. The value comprises a share, which is $O(\frac{L}{k} +  \kappa) = O(\frac{L}{t} +  \kappa) $ bits. Thus, we get the total number of bits sent from honest quorum nodes to parties is 
    \begin{align*}
        O(t) \cdot O(f) \cdot O(\log t) \cdot O(\frac{L}{t} + \kappa) = O(L\cdot f\cdot \log t + t\cdot f \cdot \log t \cdot \kappa).
    \end{align*}
\end{proof}

\begin{theorem}
    QAB has bit complexity $O(((L+\log t \cdot\kappa)\cdot n + t \cdot f\cdot \kappa) \cdot \log n \cdot \log t)$.
\end{theorem}
\begin{proof}
    Summing up the bit complexities of \cref{lem:bit complexity quorum relayers,lem:bit complexity relayers parties,lem:bit complexity parties relayers,lem:bit complexity relayers quorum,lem:bit complexity quorum parties} we obtain
    \begin{align*}
        O(L\cdot f\cdot\log t\cdot \log n + f\cdot\log t\cdot \log n\cdot \kappa) &+ O((L + \log t \cdot\kappa)\cdot n \cdot \log n \cdot \log t) \\
        &+O(n \cdot \log n \cdot \log t \cdot \kappa) \\
        &+O(t\cdot f\cdot \log n \cdot \log t \cdot \kappa)\\
        &+O(L\cdot f\cdot \log t + t\cdot f \cdot \log t \cdot \kappa)\\
        &=O(((L+\log t \cdot\kappa)\cdot n + t \cdot f\cdot \kappa) \cdot \log n \cdot \log t).
    \end{align*}
\end{proof}

\section{Improved Disperser Committees}
\label{sec:appendix:proofs}
\nodesToCommitteesPlusThm*
\begin{proof}
    We first write explicitly the constants from \cref{thm:bipartite disperser}: there exists $D \geq 1, \alpha > 0$ such that for any $n \leq m$, there exists a $(k, 0.1)$ disperser $G=(L \sqcup R)$ with $|L| = n$, $|R| = m$, left degree $d = D \log n$ and $k \geq \alpha m / \log n$.

    Let $F = 9D + 1$. We use this theorem with $n$, $m = 4F \cdot (9t+1)\cdot \log n$. Properties (1) and (2) are verified. Moreover, by setting $C = \alpha \cdot 4F \cdot 10 > 0$, property (3) is a consequence of $G$ being a $(C \cdot t, 0.1)$ disperser.
\end{proof}

\section{All to Quorum Broadcast}\label{section:AQB-appendix}

\subsection{Pseudocode}

\begin{dianabox}{\textsc{PartyAQB}}
\algoHead{AQB for Party $p$ with input $\vin$}
    \begin{algorithmic}
        \For{each $C \in \mathcal{C}$ s.t. $p \in C$}
            \State Send $hash(\vin)$ to $relayer(C)$
        \EndFor
    \end{algorithmic}
\end{dianabox}

\begin{dianabox}{\textsc{RelayerAQB}}
\algoHead{AQB for relayer $r$ of committee $C$}
    \begin{algorithmic}
        \UponTrue{received valid hash from all members of $C$}
            \State $q \gets $ the quorum node linked to $r$
            \If{All hashes are the same hash $h$}
                \State Send $h $ to $q$
            \Else 
                \State Send $\ast$ to $q$
            \EndIf
        \EndUpon
    \end{algorithmic}
\end{dianabox}

\begin{dianabox}{\textsc{QuorumNodeAQB}}
    \algoHead{AQB for Quorum Node $q$ with input $\vin$}
    \begin{algorithmic}
        \State $v \gets $ $q$'s input in consensus
        \State $R \gets$ the size of the relayer batch ($4F \cdot \log n$)
        \Statex
        \UponTrue{received more than $R_q/2$ hashes from batch for the same $h$ \textbf{and} $hash(\vin) = h$ \textbf{and} haven't decided yet} 
            \State Decide $\vin$
        \EndUpon

        \Statex
        \UponTrue{received $3 R_q / 4$ hashes \textbf{and} haven't decided yet}
            \State Decide $\ast$ 
        \EndUpon
    \end{algorithmic}
\end{dianabox}

\subsection{Analysis}

We now give the properties of our protocol for AQB. 

\begin{lemma}
    The protocol satisfies robust termination
\end{lemma}
\begin{proof}
     Assume by contradiction that $f+1$ parties never decide a value. For a party not to decide a value, it must not receive a message from $25\%$ of its batch of relayers. Because each quorum node has a batch of $4 \cdot F \cdot \log n$ relayers and each relayer is matched to a single quorum node, this means that at least $(f + 1) \cdot F \cdot \log n$ relayers did not send a message to their quorum node, which can only happen when their committee is blocked. However, there are $f$ byzantine parties, each part of $D \log n$ committees. So at most $f \cdot D \cdot \log n \leq f \cdot F \cdot \log n$ committees are blocks (because $D \leq F$ by definition). This is a contradiction.
\end{proof}

\begin{lemma}
    The protocol satisfies robust validity
\end{lemma}
\begin{proof}
    If all honest parties have the same input $v$, then every non-blocked committee (made only of honest parties) will send $hash(v)$ to its quorum node. There are at most $t$ byzantine parties, each part of $D \log n$ committees, so there are at most $D \cdot t \cdot \log n \leq F \cdot t \cdot \log n$ blocked committees. 
    We remark that for an honest quorum node node to decide $v$, given that its input is $v$, it must receive messages for a value different than $hash(v)$ from at least $25\%$ of its batch of relayers. Because only relayers for blocked committees would do so, there are at most $F \cdot t \cdot \log n$ blocked committees and each quorum node has a batch size of $4 \cdot F \cdot \log n$, there can only be at most $4 \cdot F \cdot t \cdot \log n / 4 \cdot F \cdot \log n = t$ honest quorum parties which do not decide $v$. Out of the $k$ quorum parties, given $f \leq t$ of them can be byzantine and $t$ can be honest and not decide $v$, the remaining $k - 2t$ will decide $v$.
\end{proof}

\begin{lemma}
    The protocol satisfies non-amplification.
\end{lemma}
\begin{proof}
    We use the constant $C$ defined in our committee construction from \cref{thm:nodes2committeesPlus}. Let $v \ne \ast$ be any value and assume that at least $C \cdot t$ honest parties do not have $v$ as input. 
    Assume by contradiction that $t+1$ honest quorum parties decide $v$. For a quorum party to decide $v$, it must have received $hash(v)$ from at least half of its batch, i.e it must have received $hash(v)$ from half of its $4 \cdot F \cdot \log n$ relayers. So at least $(t+1)\cdot 2 \cdot F \cdot \log n$ relayers sent $hash(v)$ as their message. As seen in the previous two lemmas, at most $t \cdot F \cdot \log n$ are from blocked committees. So at least $(t+1)\cdot F \cdot \log n$ relayers from fully honest committees sent $hash(v)$, meaning all of their committee members had $v$ as input, let $A$ be this set of committees. 
    Let $S$ be the set of honest parties whose input is not $v$, by assumption, $|S| \geq C \cdot t$. We can therefore use \cref{thm:nodes2committeesPlus} which shows that at least $90\%$ of committees have a member in $S$. So at most $10\%$ of committees have no member in $S$, let $B$ be this set, we have $|B| < 0.1 \cdot (9t+1) \cdot F \cdot \log n \leq t \cdot F \cdot \log n$ (assuming $t \geq 1$). But this is a contradiction as $A \subseteq B$ but $|A| > |B|$.
\end{proof}

\begin{lemma}
    The bit complexity of AQB is $O(n\cdot \log n\cdot \kappa)$.
\end{lemma}
\begin{proof}
    Each party sends $O(\log n)$ hashes, this is $O(n\cdot \log n\cdot \kappa)$ bits, each relayer sends $1$ hash, that is $O(t\cdot \log n \cdot \kappa)$ bits.
\end{proof}

\section{Improved QAB}\label{section:Improved-QAB-Appendix}

We now slightly modify QAB to account for the case where most of the parties hold the value decided in the Quorum Agreement. 

In particular, relayers first ask parties before sending a value, whether they have it (by sending a hash), and only send the full value in case of no response. 

\subsection{Pseudocode}

\begin{dianabox}{\textsc{ImprovedRelayerWave}}
\algoHead{Improved $\hat{f}$-Wave for Relayer of Committee $C$}
    \begin{algorithmic}
        \UponTrue{$|\mathrm{shares}| \geq k/4$} \Comment{Overwrite previous version}
            \State $v^\ast \gets decode(\mathrm{shares})$
            \If{$v^\ast = 0$}
                \State Send $\langle v^\ast, proof \rangle$ to parties in $C$
            \Else
                \State $h \gets hash(v^\ast)$
                \State Send $\langle \mathrm{''need?''}, h\rangle$ to parties in $C$
            \EndIf
        \EndUpon

        \UponTrue{receive $\langle \mathrm{''need''}\rangle$ from $p \in C$}
            \State Send $\langle v^\ast, proof\rangle$ to $p$
        \EndUpon
    \end{algorithmic}
\end{dianabox}

\begin{dianabox}{\textsc{ImprovedPartyWave}}
\algoHead{Improved $\hat{f}$-Wave for Party $p$}
    \begin{algorithmic}
        \State $v \gets$ $p$'s proposal in consensus
        \Statex
        \UponTrue{receive $\langle \mathrm{''need?''}, h\rangle$ from relayer $r$ of $C \in \mathcal{C}$ s.t. $p \in C$}
            \If{$hash(v) \neq h$}
                \State Send $\langle \mathrm{''need''}\rangle$ to $r$

            \Else
                \State $v^\ast \gets v$
                \State $\sigma \gets$ sign $h$ with aggregate signature scheme for $C$
                \State Send $\sigma$ to $r$
                \State \textbf{Decide} $v^\ast$
            \EndIf
        \EndUpon
    \end{algorithmic}
\end{dianabox}

\subsection{Analysis}

What we need to show is that now, relayers send fewer bits to parties. In particular, we improve the Lemma \ref{lem:bit complexity relayers parties}. 

\begin{lemma}[Lemma \ref{lem:bit complexity relayers parties} Improved.] 
    When using AQB before the quorum Agreement and Improved QAB after, in Improved QAB, in total, honest relayers send at most $O((L +\log t \cdot \kappa)\cdot t \cdot \log t+ n \cdot \log n \cdot \log t\cdot \kappa)$ bits to parties after GST. 
\end{lemma}
\begin{proof}
    Independent of a wave, the number of links between relayers and parties is $O(n \log n)$.
    Therefore, in each wave, relayers send $O(n\log n)$ messages of type $\langle \mathrm{''need?''}, h\rangle$, which sums to $O(n \cdot \log n \cdot \log t\cdot \kappa)$ bits over all waves. 

    If $\ast$ was decided in Quorum Agreement, relayers send  $O(n\log n)$ messages of type $\langle \ast, proof \rangle$, which sums to $O(n \cdot \log n \cdot \log t\cdot \kappa)$ bits over all waves.

    If some value $v' \neq \ast$ was decided in Quorum Agreement, by representativity and QAB's non-amplification,
    $v'$ is an input of all but $O(t)$ honest parties. Therefore, at most $O(t)$ honest parties will send a message saying they need a whole value, which means relayers will send in total $O((L +\log t \cdot \kappa)\cdot t \cdot \log t)$ bits in all waves for sharing an actual value.

    What is left to observe is that in all waves, relayers might send full value to all the malicious parties, which sums up to $O((L + \log t \cdot \kappa)\cdot f\cdot\log t)$.

    Summing all up, the total complexity is $O((L +\log t \cdot \kappa)\cdot t \cdot \log t+ n \cdot \log n \cdot \log t\cdot \kappa)$ as stated.
\end{proof}

\section{Partially Synchronous Retrieval Protocol}\label{section:gst-retrieval-appendix}

\subsection{Pseudocode}

\begin{dianabox}{\textsc{RetrievalLeaderGST}}
\algoHead{Retrieval protocol for the leader in synchrony}
\begin{algorithmic}
    \State Wait for $2\Delta$ units of time
    \If{there exists $\textbf{acc}$ which received at least $t+1$ partial signatures $(\rho_p)$ with valid shares $(s_p)$}
        \State $\textbf{proof} \gets tcombine(\textbf{acc}, (\rho_p))$
        \State $v \gets decode((s_p))$
        \State Send $(INFORM, \textbf{acc}, v)$ to every party 
        \State Wait for $n-t$ partial signatures $(\rho'_p)$ for $(RECEIVED, \textbf{acc})$
        \State $\textbf{proof'} \gets tcombine((\rho'_p), (RECEIVED, \textbf{acc}))$
        \State \Return $(AGREE,  \textbf{value}, \textbf{proof}, \textbf{proof'})$
    \EndIf
    
    \State Let $S$ be the set of hashes received
    \State Let $\mathcal{I}$ be a partition of the hash range using \cref{lemma:input-inter-partition}
    \State Broadcast $(RANGES, \mathcal{I})$ to every party
    \State Wait for $t+1$ valid threshold signatures $(\rho_{I,p})$ for every $I \in \mathcal{I}$
    \State For every $I \in \mathcal{I}$, $\textbf{proof}_I \gets tcombine(I, (\rho_{I,p}))$
    \State \Return $(DISAGREE, I, (\textbf{proof}_I)_{I \in \mathcal{I}})$
\end{algorithmic}
\end{dianabox}

\begin{dianabox}{\textsc{RetrievalPartyGST}($\vin$)}
\algoHead{Retrieval protocol for a party $p_i$ in synchrony}
\begin{algorithmic}
    \State $(acc, (s_j)) \gets encode(\vin)$
    \State Send $(acc, tsign(acc),  s_i)$ to the leader
    \UponTrue{Received $(RANGES, \mathcal{I})$ with $|I| \leq 7$ from the leader}
        \State $R \gets \emptyset$
        \For{$I \in \mathcal{I}$ such that $acc \notin I$}
            \State $R \gets R \cup \{(I, tsign(I))\}$
        \EndFor
        \State Send $R$ to the leader
    \EndUpon
    \UponTrue{Received valid $(INFORM, \textbf{acc}, v)$ from the leader}
        \State Store $v$
        \State $\rho \gets tsign((RECEIVED, acc))$
        \State Send $\rho$ to the leader
    \EndUpon

\end{algorithmic}
\end{dianabox}

\subsection{Analysis}

\begin{lemma}\label{lemma:gst-retrieval-agreement}
    If the retrieval protocol returns an agreement evidence for an accumulator $acc$, then the value $v$ associated with $acc$ satisfies strong unanimity. Moreover, at least $t+1$ honest parties know $v$.
\end{lemma}

\begin{proof}
    The first part of the agreement proof is a $t+1$ threshold signature that an honest party only sign if its input value is $v$ (assuming no collision on the accumulator value). Therefore, at least one honest party has $v$ as input so this value satisfies strong unanimity. Moreover, the second part of the agreement proof is a $n-t$ threshold signature signed by parties which know $v$. This implies that at least $n - t - t \geq t + 1$ honest parties know $v$.
\end{proof}

\begin{lemma}\label{lemma:gst-retrieval-disagreement}
    If the retrieval protocol returns a disagreement evidence, then not all honest parties have the same input.
\end{lemma}
\begin{proof}
    A disagreement evidence consists of a partition of the root range $\mathcal{I}$ as well as for each interval $I \in \mathcal{I}$ a $(t+1)$-threshold that parties only sign if their input is not in $I$. Assume by contradiction that all honest parties have the same input, let $I$ be the interval in $\mathcal{I}$ containing this input. Then no honest party will vote for this interval, meaning the threshold signature could not get formed. Therefore all honest parties do not have the same input.
\end{proof}

\begin{lemma}\label{lemma:gst-honest-leader}
    After $GST$, when the leader is honest, the retrieval protocol will always return a proof within $4$ rounds.
\end{lemma}

\begin{proof}
    Assume we are after GST and the leader is honest. If the leader receives any valid value $t+1$ times, it will be able to get a threshold signature on it and will return the agreement proof within $4$ rounds.

    In the other case, it means every value appeared at most $t < n/3$ times. The interval partition $\mathcal{I}$ is done using \cref{lemma:input-inter-partition}. This implies for any interval $I \in \mathcal{I}$ that at most $t$ honest parties have their input in $I$, so at least $n - t - t \geq t+1$ honest parties have their input outside $I$, so they will send partial signature for it and the leader will be able to get its disagreement evidence within $4$ rounds. 
\end{proof}

\begin{lemma}\label{lemma:gst-retrieval-bit-complexity}
    After GST, non-leader parties send together at most $\mathcal{O}(n \cdot \kappa + L)$ bits while the leader sends at most $\mathcal{O}(n \cdot(L + \kappa))$ bits.
\end{lemma}

\begin{proof}
    We remark that each of the $n-1$ non-leader parties send at most one share and accumulator, of size $\mathcal{O}(n/L + \kappa)$ and a constant amount of signatures of size $\mathcal{O}(\kappa)$ to the leader. So the total message complexity for non-leaders is $\mathcal{O}(L + n \cdot \kappa)$.

    Meanwhile, we remark that $\mathcal{I}$ consists of at most $7$ intervals, which can each be represented with $2\kappa$ bits. Otherwise the leader sends a constant amount of signatures to each party and, in the case of an agreement proof, the whole value. This takes in total $\mathcal{O}(n\cdot (L + \kappa))$ bits.
\end{proof}

\section{Synchronous Retrieval Protocol}\label{section:sync-retrieval-appendix}

\subsection{Pseudocode}

\begin{dianabox}{\textsc{RetrievalLeaderSync}}
\algoHead{Retrieval protocol for the leader in synchrony}
\begin{algorithmic}
    \State Wait for $2\Delta$ units of time
    \State Let $S$ be the multiset of hashes received
    \State Let $A$ be the elements of $S$ which occur at least $n/4$ times
    \State Let $\mathcal{I}$ be a partition of the hash range with $S \setminus A$ using \cref{lemma:input-inter-partition}
    \State Broadcast $(SIGN, A, \mathcal{I})$ to every party
    \State Wait $2\Delta$ units of time
    \If{there exists $(ACCEPT, \textbf{acc})$ which received at least $t+1$ partial signatures $(\rho_p)$ and at least $n/4$ valid shares $(s_p)$}
        \State $\textbf{proof} \gets tcombine(\textbf{acc}, (\rho_p))$
        \State $v \gets decode((s_p))$
        \State Send $(INFORM, \textbf{acc}, v)$ to every party 
        \State Wait for $t+1$ partial signatures $(\rho'_p)$ for $(RECEIVED, \textbf{acc})$
        \State $\textbf{proof'} \gets tcombine((\rho'_p), (RECEIVED, \textbf{acc}))$
        \State \Return $(AGREE,  \textbf{acc}, \textbf{proof}, \textbf{proof'})$
    \EndIf
    
    \State Set party's own input to $\bot$
    \If{Every element in $A$ and $\mathcal{I}$ received at least $t+1$ partial signatures}
        \State Combine all partial signatures into threshold signatures
        \State Make a disagreement evidence \textbf{evidence}
        \State \Return $(DISAGREE, \textbf{evidence})$
    \Else
        \State \Return $\bot$
    \EndIf
\end{algorithmic}
\end{dianabox}

\begin{dianabox}{\textsc{RetrievalPartySync}($\vin$)}
\algoHead{Retrieval protocol for a party $p_i$ in synchrony}
\begin{algorithmic}
    \If{$\vin \ne \bot$}
        \State $(acc, (s_j)) \gets encode(\vin)$
        \State Send $acc$ to the leader
    \Else 
        \State $root \gets \bot$
    \EndIf
    \UponTrue{Received $(A, \mathcal{I})$ with $|A| \leq 4$ and $|\mathcal{I}| \leq 9$ from the leader}
        \State $R \gets \emptyset$
        \For{$I \in \mathcal{I}$ such that $\vin \notin I$}
            \State $R \gets R \cup \{(I, tsign(I))\}$
        \EndFor
        \For{$a \in A$}
            \If{$acc \ne a$}
                \State $R \gets R \cup \{(a, tsign(a))\}$
            \EndIf
            
            \If{$acc = a$}
                \State $R \gets ((ACCEPT, a), tsign((ACCEPT, a)), s_i)$
            \ElsIf{$acc = \bot$}
                \State $R \gets ((ACCEPT, a), tsign((ACCEPT, a)), \bot)$
            \EndIf
        \EndFor
        \State Send $R$ to the leader
    \EndUpon
    \UponTrue{Received valid $(INFORM, \textbf{acc}, v)$ from the leader}
        \State Store $v$
        \State $\rho \gets tsign((RECEIVED, \textbf{acc}))$
        \State Send $\rho$ to the leader
    \EndUpon

\end{algorithmic}
\end{dianabox}

\subsection{Analysis}

\begin{lemma}\label{lemma:set-to-bot}
    If an honest party sets its input to $\bot$, then all honest parties do not have the same input.
\end{lemma}

\begin{proof}
    Consider the first occurrence where an honest party $p$ sets its input to $\bot$. It does so as a leader after failing to retrieve $t+1$ partial signatures and $n/4$ valid shares for a value. By contradiction, assume that all honest parties agree on the same value $v$ and let $acc$ be its accumulator. Because there are at least $n-t \geq t + 1 \geq n/4$ honest parties, the leader must have received $acc$ at least $n/4$ times and thus added it to $A$. This means that every honest party would have sent a threshold signature and share for $v$, hence the leader would have been able to get an agreement evidence and would not set its input to $\bot$, hence the contradiction.
\end{proof}

\begin{lemma}\label{lemma:sync-retrieval-parties-know}
    If the retrieval protocol returns an agreement evidence for a value $v$, then it satisfies strong unanimity. Moreover, at least $t+1-f$ honest parties know $v$.
\end{lemma}

\begin{proof}
    The first part of the agreement evidence is a $t+1$ threshold signature that an honest party only sign if its input value is $v$ or $\bot$. Therefore, at least one honest party has $v$ or $\bot$ as input. If the input is $v$, this value satisfies strong unanimity. If the input it $\bot$, using \cref{lemma:set-to-bot}, the value also satisfies strong unanimity. Moreover, the second part of the agreement proof is a $t+1$ threshold signature signed by parties which know $v$. This implies that at least $t+1-f$ honest parties know $v$.
\end{proof}

\begin{lemma}\label{lemma:sync-retrieval-disagreement}
    If the retrieval protocol returns a disagreement evidence, then all honest parties do not have the same input.
\end{lemma}
\begin{proof}
    Assume by contradiction that all honest parties have the same input with accumulator $acc$ and the retrieval protocol returned a disagreement evidence with sets $A$ and $\mathcal{I}$. Using \cref{lemma:set-to-bot}, no party set its input to $\bot$. We have two cases:
    \begin{itemize}
        \item $acc \in A$, then no honest party will send a partial signature for the element $acc$.
        \item $acc \notin A$. Then, because $\mathcal{I}$ partitions the input space, we can find $I \in \mathcal{I}$ such that $acc \in I$. Thus $acc \in I \setminus A$, so no honest party will send a partial signature for $I$.
    \end{itemize}
    In both cases, no honest party will send a partial signature for an element in the disagreement signature, meaning it won't be possible to make a $(t+1)$-threshold signature for one of its element, hence  the disagreement evidence is not valid.
\end{proof}

\begin{lemma}\label{lemma:sync-retrieval-run-different-leaders}
    If $f < n/4$, when run by $f+1$ different honest leaders, at least one honest leader will get an agreement or disagreement evidence.
\end{lemma}

\begin{proof}
    Assume that $f < n/4$ and the protocol was already run by $f$ honest leaders and always returned $\bot$, we want to show that the next honest leader will always return a proof. We note that when a leader returns bot, it always sets its input to $\bot$, so $f$ parties set their input to $\bot$.

    Let $a \in A$, we want to show that the leader will be able to get a $(t+1)$-threshold signature for having $a$ with at least $n/4$ shares or a $(t+1)$-threshold signature for not having $a$. Let $nb_1$ be the honest parties which would send a partial signature saying they have $a$ and $nb_0$ the honest parties which would send a partial signature saying they don't have $a$. Before the retrieval protocol is run for the first time, we have $nb_1 + nb_0 = n-f$ (every honest party sends a partial signature for at least one of them). We remark that if a party sets its input to $\bot$, then it will send a partial signature for both, meaning $nb_1 + nb_0$ will increase by $1$. Therefore, after $f$ distinct honest parties set their input to $\bot$, we have $nb_1 + nb_0 \geq n - f + f = n = 2t + 1$. Therefore, using the pigeonhole principle, either $nb_0 \geq t+1$ in which case the leader can get a threshold signature for $a$ for the proof of disagreement, either $nb_1 \geq t+1$. In this latter case, we remark that $f$ honest parties have set their input to $\bot$ so far, which implies that $nb_1 - f$ parties have accumulator $acc$ and thus sent a share along with the partial signature. We have $nb_1 - f \geq t + 1 - n/4 \geq n/4$. Hence, the leader will be able to reconstruct the input, get the threshold signature for it and get an agreement evidence for it.

    Let $I \in \mathcal{I}$, we want to show the leader will be able to get a $(t+1)$-threshold signature for $I$. We remark that by construction of $I$, less than $n/4$ honest parties have their accumulator contained in $I \setminus A$. Hence the leader will receive at least $n - f - n/4$ partial signatures from honest parties. Because we assumed $f < n/4$, the leader receives strictly more than $n - n/4 - n/4$ partial signatures, so at least $t+1$ partial signatures and thus will be able to make a threshold signature out of it.

    As a consequence, either the leader will be able to get an agreement evidence, or it will get all threshold signatures necessary to get a disagreement evidence.
\end{proof}

\begin{lemma}\label{lemma:sync-leader-bits}
    Non-leader parties send together at most $\mathcal{O}(n \cdot \kappa + L)$ bits while the leader sends at most $\mathcal{O}(n \cdot(L + \kappa))$ bits. Moreover, if the leader is honest and $f < n/4$, then the retrieval protocol terminates in at most $6$ rounds.
\end{lemma}

\begin{proof}
    We remark that each of the $n-1$ non-leader parties send at most one share and accumulator, of size $\mathcal{O}(n/L + \kappa)$ and a constant amount of signatures of size $\mathcal{O}(\kappa)$ to the leader. So the total message complexity for non-leaders is $\mathcal{O}(L + n \cdot \kappa)$.

    Meanwhile, we remark that $\mathcal{I}$ consists of at most $9$ intervals, which can each be represented with $2\kappa$ bits. Otherwise the leader sends a constant amount of signatures to each party and, in the case of an agreement proof, the whole value. This takes in total $\mathcal{O}(n\cdot (L + \kappa))$ bits.

    Regarding round complexity, in the worst case the protocol consists of $3$ back-and-forths between the leader and parties, so it terminates within $6$ rounds.
\end{proof}

\section{Partially Synchronous Quorum Agreement Analysis}\label{section:gst-quorum-appendix}

\begin{lemma}\label{lemma:gst-quorum-round}
    \textsc{ViewByzantineAgreement} returns the same agreement or disagreement evidence to all honest parties within $\mathcal{O}(f)$ rounds after GST, and $proof$ is a certificate for it.
\end{lemma}

\begin{proof}
    This is a consequence of \cite[Lemma B.3, Corollary B.8]{constantinescu2025few}. To be more specific, $proof$ is a commit proof for $value$ so \cite[lemma B.3]{constantinescu2025few} guarantees that all honest parties will agree on it. And \cite[Corollary B.8]{constantinescu2025few} guarantees that \textsc{ViewByzantineAgreement} returns within $\mathcal{O}(f)$ rounds. We note that the value passed agreed on in the protocol is always a non-$\bot$ output of the retrieval protocol, meaning using our retrieval protocol it is either an agreement or disagreement evidence.
\end{proof}

\begin{lemma}\label{lemma:gst-quorum-validiy}
    If a party decides a value, then it satisfies strong unanimity.
\end{lemma}

\begin{proof}
    If a party decides $\ast$, then it receives a disagreement proof, which means using \cref{lemma:gst-retrieval-disagreement} that $\ast$ can be decided. It decides a value $v$, then it receives an agreement proof for $v$, which implies using \cref{lemma:gst-retrieval-agreement} that $v$ satisfies strong unanimity.
\end{proof}

\begin{lemma}\label{lemma:gst-quorum-partial-termination}
    \textsc{QuorumGST} satisfies partial termination in $\mathcal{O}(f)$ rounds.
\end{lemma}

\begin{proof}
    As proven above, all honest parties receive the same agreement or disagreement proof within $\mathcal{O}(f)$ rounds. If it is a disagreement proof, all parties will decide $\bot$, because there are $n-t \geq 2n/3 \geq n/4$ honest parties, it satisfies partial termination.

    If it is an agreement proof, using \cref{lemma:gst-retrieval-agreement}, at least $t+1 \geq n/3 \geq n/4$ honest parties got the full value matching this accumulator in an $INFORM$ message and therefore will decide it, so partial termination is also satisfied.
\end{proof}

\begin{lemma}\label{lemma:gst-quorum-bit}
    \textsc{QuorumGST} has $\mathcal{O}(n \cdot(L + f\cdot \kappa))$ bit complexity. 
\end{lemma}

\begin{proof}
    We note that outside the retrieval protocol, the word complexity of our protocol is the same as \textsc{ViewByzantineAgreement} which was proven to be $\mathcal{O}(n\cdot f)$ \cite[lemma B.11]{constantinescu2025few} (We note that because $t < n/3$, then we always have $f \leq n - k$ where $k = \lceil \frac{n + t + 1}{2} \rceil$). Because the authors of \cite{constantinescu2025few} consider a word to be $\mathcal{O}(\kappa)$, the total bit complexity outside of the retrieval protocol is $\mathcal{O}(n \cdot f \cdot \kappa)$. 

    We now consider the retrieval protocol. Using \cref{lemma:gst-retrieval-bit-complexity}, when used, the retrieval protocol has bit complexity $\mathcal{O}(n(L +  \kappa))$ if the leader is honest and $\mathcal{O}(L + \kappa)$ otherwise. We remark that using \cite[Lemma B.6]{constantinescu2025few}, if after GST, a view starts with an honest leader (which happens after at most $f+1$ views), then \cref{lemma:gst-honest-leader} ensures that the retrieval protocol will always return a proof, so a commit proof will be created and shared to every honest party, thus the retrieval protocol won't be run anymore. Hence, after GST, the retrieval protocol will be run at most $f+1$ times and at most $1$ time with an honest leader. This has total bit complexity $\mathcal{O}(n \cdot(L + f\cdot \kappa))$. We remark that if GST happens during a view, then the retrieval protocol can fail while still using $\mathcal{O}(n(L +  \kappa))$ bits. However, this can happen at most once.

    Summing the bit complexity of the retrieval protocols invocations and the rest of the protocol, we get in total the expected $\mathcal{O}(n \cdot(L + f\cdot \kappa))$ bit complexity.
\end{proof}

\begin{lemma}
    \textsc{QuorumGST} satisfies 
\end{lemma}

\begin{theorem}
    \textsc{QuorumGST} satisfies agreement, strong unanimity, provability and partial termination in $\mathcal{O}(f)$ rounds.
\end{theorem}

\begin{proof}
    This is a consequence of \cref{lemma:gst-quorum-round,lemma:gst-quorum-validiy,lemma:gst-quorum-partial-termination,lemma:gst-quorum-bit}.
\end{proof}

\section{Synchronous Quorum Agreement}\label{section:sync-quorum-appendix}

\subsection{Pseudocode}

We provide below our implementation of an adaptive protocol for the synchronous setting with $\mathcal{O}(n\cdot( L + f \cdot \kappa))$ bit complexity and $\mathcal{O}(f + \log n)$ round complexity. We note that this protocol is mostly similar to the synchronous protocol for binary agreement given by \cite[Appendix B]{constantinescu2025few}. As such, our modifications to this protocol are given in blue. This protocol uses the parameter $k = \lceil \frac{n+t+1}{2}\rceil \approx 3n/4$. $c$ is the constant such that there exists a component of size $n - \mathcal{O}(f)$ in \cref{theo:sync-expander}.

For a party $p$, we say that $p$ resolves a retrieval proof if it is a disagreement proof (in which case it is resolved to $\ast$) or is is an agreement proof with accumulator $acc$ and $p$ knows the value $v$ with such accumulator (in which case it is resolved to $v$).

\begin{dianabox}{\textsc{ViewByzantineAgreementSync}($\vin$)}
\algoHead{Protocol for Byzantine Agreement in synchrony for a party $p_i$}
\begin{algorithmic}
\State $\textbf{retrieval\_value} \gets \bot$
\State $\textbf{key} \gets \bot$
\State $\textbf{lock} \gets \bot$
\State $\textbf{commit} \gets \bot$
\State $\textbf{dispersal} \gets \bot$
\State $\textbf{output} \gets \bot$

\Statex{$\triangleright$  Decide as soon as we get a dispersion value that can be resolved}
\UponTrue{Receiving valid $\textsc{SendDispersal}(\textbf{value}, \textbf{proof})$}
    \If{$\textbf{dispersal} = \bot$}
        \State $\textbf{dispersal} \gets (\textbf{value}, \textbf{proof})$
        \If {$p_i$ has not yet decided and $\textbf{value}$ can be resolved to $v$}
            \State Decide $(v, \textbf{proof})$
        \EndIf
    \EndIf
\EndUpon

\Statex{$\triangleright$  Ask for shares $\mathcal{O}(\log n)$ after receiving a dispersion proof}
\UponSimple{$c \cdot \log n$ rounds after receiving a dispersal proof, if $p$ has not yet decided}
    \State Send $(\textsc{RequestShare}, \textbf{dispersal})$ to every party
\EndUpon

\Statex{$\triangleright$  Send share to every party requesting it}
\UponTrue{Receiving $(\textsc{RequestShare}, proof)$ from a party $p_j$ for the first time}
    \If{$proof$ resolves to a value $v$}
        \State $(acc, (s_j)) \gets encode(v)$
        \State Send $(\textsc{SendShare}, s_i)$ to $p_j$
    \EndIf
\EndUpon

\Statex{$\triangleright$ Reconstruct value using shares}
\UponTrue{Receiving $(\textsc{SendShare}, s_j)$ from $p_j$ where $s_j$ is valid for \textbf{dispersion}}
    \State Save $s_j$
    \If{$p_i$ has not yet decided and received $n/4$ shares $(s_j)$ for the accumulator in \textbf{dispersion}}
        \State $v \gets decode((s_j))$
        \State Decide $(v, \textbf{dispersion})$
    \EndIf
\EndUpon

\Statex{$\triangleright$ Disperse the decided value using the expander}
\UponTrue{Receiving valid $(\textsc{Disperse}, v)$}
    \If{$\textbf{dispersal}$ is a dispersal proof for $v$}
        \State Decide $(disp, commit)$
    \EndIf
    \State Forward the message to $p_i$'s neighbors in $G$
\EndUpon
        
\Statex{$\triangleright$  Run the view-based protocol}
\For{view number $view \gets 0,1, \ldots$}
    \If{This party is the current leader}
        \If{$\textbf{dispersal} = \bot$}
            \State Run ViewLeaderProtocol() in parallel for duration $11\Delta$
        \Else{}
            \UponSimple{Receiving $(\textsc{Complain})$ from a party $v$ for the first time}
                \State Send $\textsc{SendDispersal}(\textbf{dispersal})$ to $v$
            \EndUpon
        \EndIf
    \EndIf
    \State Run ViewPartyProtocol($\vin$, $view$, \textbf{key}, \textbf{lock}, \textbf{commit}) for duration $11\Delta$
\EndFor

\end{algorithmic}
\end{dianabox}

\begin{dianabox}{\textsc{ViewLeaderProtocolSync}}
\algoHead{View-based protocol part exclusive to the leader in synchrony}

\begin{algorithmic}[1]
        \Statex{$\triangleright$  Choose which value to propose}
        \State Broadcast \textsc{RequestSuggestion}
        \State Wait for valid \textsc{Suggest}(m) from $k$ parties
        \If{one of the value is a dispersion proof $disp$ with value $v$:}
            \State $\textbf{value} \gets v$
            \State $\textbf{dispersion} \gets disp$
            \State Jump to broadcasting the dispersion value
        \ElsIf{One of the value is a key value:}
            \State Let $key$ be the key with the highest view number and $v$ its value
            \State $\textbf{value} \gets v$
            \State $\textbf{prop} = (KEY, key)$
        \ElsIf{One of the value is a retrieval proof $proof$:}
            \State $\textbf{prop} = (COMBINE, proof)$
        \Else
            \State Broadcast \textsc{RunRetrieval}
            \State $(\textbf{value}, \textbf{proof}) \gets \textsc{RetrievalLeader}()$
            \If{$\textbf{value} = \bot$}
                \State No value to propose
                \State Stay silent for the rest of the round
            \EndIf
            \State $\textbf{retrieval\_value} \gets (\textbf{value}, \textbf{proof})$
            \State $\textbf{prop} = (COMBINE, proof)$
        \EndIf

        \Statex
        \Statex{$\triangleright$  Approve the proposed value and get a key}
        \State Send \textsc{ProposeKey}($\textbf{value}, \textbf{prop}$) to every party.
        \State Wait for valid \textsc{CheckedKey}($\rho_p$) from $k$ parties
        \State $\textbf{key\_proof} \gets tcombine((KEY, \textbf{value}), (\rho_p))$

        \Statex
        \Statex{$\triangleright$ Approve the key and get a lock}
        \State Send \textsc{ProposeLock}($\textbf{value}, \textbf{key\_proof}$) to every party.
        \State Wait for valid \textsc{CheckedLock}($\rho_p$) from $k$ parties
        \State $\textbf{lock\_proof} \gets tcombine((LOCK, \textbf{value}), (\rho_p))$

        \Statex
        \Statex{$\triangleright$ Approve the lock and get a commit}
        \State Send \textsc{ProposeCommit}($\textbf{value}, \textbf{lock\_proof}$) to every party.
        \State Wait for valid \textsc{CheckedCommit}($\rho_p$) from $k$ parties
        \State $\textbf{commit\_proof} \gets tcombine((COMMIT, \textbf{value}), (\rho_p))$

        \Statex
        \Statex{$\triangleright$ Approve the commit and get a dispersal}
        \State Send \textsc{ProposeDispersal}($\textbf{value}, \textbf{commit\_proof}$) to every party.
        \State Wait for valid \textsc{CheckedDispersal}($\rho_p$) from $k$ parties
        \State $\textbf{dispersal\_proof} \gets tcombine((DISPERSAL, \textbf{value}), (\rho_p))$
        
        \Statex
        \Statex{$\triangleright$ Broadcast the dispersal value}
        \State Send \textsc{SendDispersal}($\textbf{value}, \textbf{dispersal\_proof}$) to every party.
\end{algorithmic}
\end{dianabox}

\begin{dianabox}{\textsc{ViewPartyProtocolSync}(input value $\vin$, view number $view$, \textbf{key}, \textbf{lock}, \textbf{commit})}
\algoHead{View-based protocol part exclusive to parties, for a party $p$}

\begin{algorithmic}
        \Statex{$\triangleright$  Complain if you do not have a dispersal proof}
        \If{$\textbf{dispersal} = \bot$}
            \State Send $(\textsc{Complain})$ to the leader
         \EndIf

        \Statex
        \Statex{$\triangleright$  Choose which value to propose}
        \State Wait for a message \textsc{RequestSuggestion} from the leader
        \If{$\textbf{dispersal} \neq \bot$}
            \If{This party never sent the dispersal value to the leader}
                \State Send \textsc{Suggest}($(DISPERSAL, \textbf{dispersal})$) to the leader
            \EndIf
            \State Wait for the rest of the view
        \ElsIf{$\textbf{key} \neq \bot$}
             \State Send $\textsc{Suggest}((KEY, \textbf{key}))$ to the leader
        \ElsIf{$\textbf{retrieval\_value} \ne \bot$}
            \State Send $\textsc{Suggest}((VALUE, \textbf{retrieval\_value}))$ to the leader
        \Else
            \State Send $\textsc{Suggest}(())$ to the leader
        \EndIf

        \Statex
        \Statex{$\triangleright$  Run the retrieval procedure if needed}
        \Upon{Receiving from leader}{\textsc{RunRetrieval}}{}
            \State Run $\textsc{RetrievalParty}(v_{IN}$)
        \EndUpon
        
        \Statex{$\triangleright$  Check and sign the key}
        \State Wait for a valid \textsc{ProposeKey}($v$, ($type, proof$)) from the leader
        \If{$\textbf{lock} \ne \bot$}
            \If{$type = COMBINE$ or $type = KEY$ with a view number strictly less than the $\textbf{lock}$ view number}
                \State Wait for the rest of the view
            \EndIf
        \EndIf
        \State Send \textsc{CheckedKey}($tsign((KEY, v))$) to the leader
        
        \Statex
        \Statex{$\triangleright$  Check and sign the lock value}
        \State Wait for a valid \textsc{ProposeLock}($v$, $proof$) from the leader
        \State $\textbf{key} \gets (v, view, proof)$
        \State Send \textsc{CheckedLock}($tsign((LOCK, v))$) to the leader
        
        \Statex
        \Statex{$\triangleright$  Check and sign the commit value}
        \State Wait for a valid \textsc{ProposeCommit}($v$, $proof$) from the leader
        \State $\textbf{lock} \gets (v, view, proof)$
        \State Send \textsc{CheckedCommit}($tsign((COMMIT, v))$) to the leader
        
        \Statex
        \Statex{$\triangleright$  Check, disperse and sign the commit value}
        \State Wait for a valid \textsc{ProposeDispersal}($v$, $proof$) from the leader
        \State $\textbf{commit} \gets (v, view, proof)$
        \If{$v$ is an agreement proof and resolves to an input $input$}
            \If{$p$ did not transmit $input$ in $G$ yet}
                \State Send $(\textsc{Disperse}, commit, input)$ to $p$'s neighbors in $G$
            \EndIf
        \EndIf
        \State Send \textsc{CheckedDispersal}($tsign((DISPERSAL, v))$) to the leader
\end{algorithmic}
\end{dianabox}

\begin{dianabox}{\textsc{SyncBA}}
\algoHead{Optimal synchronous agreement protocol for a party $p_i$}
\begin{algorithmic}[1]
    \State $use\_fallback \gets 0$
    \For{$n$ views}
        \State Run $\textsc{ViewByzantineAgreement}$
    \EndFor
    \If{$\textbf{dispersal} \ne \bot$}
        \State $proof \gets \textbf{dispersal}$
    \ElsIf{$\textbf{commit} \ne \bot$}
        \State $proof \gets \textbf{commit}$
    \Else
        \State $proof \gets \bot$
    \EndIf
    \State Let $(acc, (s_i)) \gets decode(value)$ if $value \ne \bot$
    \Statex
    
    \If{No value was decided}
        \State Broadcast $(HELP)$ to every party
    \EndIf
    
    \Statex
    \UponSimple{Receiving $(HELP)$ from party $p_j$}
        \If{$value \ne \bot$}
            \State Send $(PROOF, proof)$ to $p$
        \EndIf
    \EndUpon
    \UponTrue{Receiving $(HELP)$ from $t+1$ different parties}
        \State $use\_fallback \gets 1$
    \EndUpon

    \Statex
    \If{Received a dispersal proof $disp$}
        \State $proof \gets disp$
    \ElsIf{Received a commit proof $commit$}
        \State $proof \gets commit$
    \EndIf
    \If{No value was decided and $proof \ne \bot$}
        \State Send $(ASK\_SHARE, proof)$ to every party
    \EndIf

    \Statex
    \UponSimple{Receiving valid $(ASK\_SHARE, proof)$ from $p_j$ if $proof$ can be resolved to a value $o$}
        \State $(acc, (s_j)) \gets encode(o)$
        \State Send $(SHARES, s_i, s_j)$ to $p_j$
    \EndUpon

    \Statex
    \UponSimple{Receiving valid $(SHARES, s_j, s_i)$ from a party $p_j$}
        \State Save $s_i$ and $s_j$
        \State Broadcast $(SHARE, s_i)$ to every party
        \State $commit\_proof \gets proof$
    \EndUpon
        
    \Statex
    \UponSimple{Receiving valid $(SHARE, s_j)$ from a party $p_j$}
        \State Save $s_j$
        \State Broadcast $(SHARE, s_i)$ to every party
        \State $commit\_proof \gets proof$
    \EndUpon
    \If{Already decided a value $o$}
        \State $output \gets o$
    \ElsIf{Got $n/4$ shares $(s_j)$ for the value associated with $proof$}
        \State $output \gets decode((s_j))$
    \Else{}
        \State $output \gets \vin$
    \EndIf
    \If{$use\_fallback = 1$}
        \State $output' \gets \mathcal{A}_{fallback}(output)$
        \If{$proof$ is not a dispersal proof, or is a dispersal proof for $output$}
            \State Broadcast $tsign((CERTIFY, output'))$ to every party
        \EndIf
        \If{$proof$ is not a dispersal proof}
            \State $output \gets output'$
            \State Wait for $t+1$ partial signatures $(\rho_p)$ for $(CERTIFY, output)$
            \State $proof \gets tcombine((\rho_p))$
        \EndIf
    \EndIf
    \If{The party has not yet decided}
        \State Decide $(output, proof)$
    \EndIf
\end{algorithmic}
\end{dianabox}

\subsection{Analysis}

We want our protocol to satisfy certified agreement and remark that in our implementation, the certificate can be of two sorts: either a dispersal proof obtained during the view-based part of the protocol or a $(t+1)$-threshold signature for $(CERTIFY, output)$ obtained. We show in \cref{lemma:sync-agreement-on-dispersal,lemma:sync-certify-prop} that this proof satisfies Provability.

We first consider the view-based part of the protocol. We start by giving results that come from the original paper, as the key, lock and commit proof structure is the same:

\begin{lemma}[Lemma B.3 of \cite{constantinescu2025few}]
    During an entire execution, at most one value will get commit proofs.
\end{lemma}

Given that a dispersal proof is built on top of an existing commit proof, we get the following corollary:
\begin{corollary}\label{coro:sync-view-commit-disp-same}
    During an entire execution, at most one value will get commit and dispersal proofs.
\end{corollary}

We now consider validity:
\begin{lemma}
    If a party decides a value during \textsc{ViewByzantineAgreementSync}, then this value satisfies strong unanimity.
\end{lemma}

\begin{proof}
    The proof is the same as \cite[Lemma B.4]{constantinescu2025few}. A value is only decided if it comes with a dispersal proof. However all dispersal proofs are originally layered on top of a proof returned by the retrieval protocol. But using \cref{lemma:sync-retrieval-parties-know,lemma:sync-retrieval-disagreement}, this value must satisfy strong unanimity.
\end{proof}

We will now look at the round complexity within the view-based protocol:
\begin{lemma}\label{lemma:sync-view-leader-bot}
    If $f < n - k$ and an honest leader get a non-$\bot$ output from the retrieval protocol, then all parties will get a dispersal proof within the view.
\end{lemma}
\begin{proof}
    If an honest leader gets a non-$\bot$ output, because there are $k$ honest parties, it will always be able to get all the $k$-partial signatures required for the key, lock, commit and dispersal proof and therefore get a dispersal proof.
\end{proof}

\begin{lemma}\label{lemma:sync-view-get-dispersal}
    If $f < \min(n/4, n-k)$, then all honest parties get a dispersal proofs within $2f + 1$ views.
\end{lemma}

\begin{proof}
    If $f < n/4$, then among $f + 1$ calls by distinct honest leaders, using \cref{lemma:sync-retrieval-run-different-leaders}, it will at least once return a non-$\bot$ value. Because among $2f + 1$ views, there are at least $f + 1$ honest leaders (and $f < n/2$ so they are distinct), one of them will get a non-$\bot$ output and using \cref{lemma:sync-view-leader-bot}, all parties will get the dispersal proof.

    We note that the statement above only holds if all the $f+1$ honest leaders run the retrieval protocol. An honest leader will not run the retrieval protocol if it receives an existing retrieval proof or lock commit, but in this case, because $f < n-k$, it will still be able to get a new key then lock, commit and dispersal threshold signatures and top of it and share it to every party. The only other case where an honest leader won't run the retrieval protocol is if it already have a dispersal proof. But in this case, all the other parties which did not have the retrieval proof yet will send a \textsc{Complain} message and get it within the view.
\end{proof}

\begin{lemma}\label{lemma:sync-view-dispersal-partial}
    If $f < n/4$ and a dispersal proof exists, at least $n/4$ honest parties know the value associated with this dispersal proof.
\end{lemma}

\begin{proof}
    If the dispersal proof is for a disagreement evidence, then all honest parties know the value (which is $\bot$). If it is an agreement evidence, using \cref{lemma:sync-retrieval-parties-know}, at least $t+1-f$ parties know the value. We remark that we assume optimal resiliency so $t = \lfloor (n-1)/2\rfloor$, thus $t+1 \geq n/2$ and because $f < n/4$, $t+1 - f \geq n/2 - n/4 \geq n/4$. So these $n/4$ honest parties know the value.
\end{proof}

\begin{lemma}\label{lemma:sync-view-partial}
    If $f < min(n/4, n - k)$, then at least $n/4$ honest parties decide within $2f+1$ views.
\end{lemma}
\begin{proof}
    Using \cref{lemma:sync-view-get-dispersal}, all honest parties will get a dispersal proof within $\mathcal{O}(f)$ rounds. Because $f < n/4$, using \cref{lemma:sync-view-dispersal-partial} at least $n/4$ honest parties already know the value and thus decide it.
\end{proof}

For the following lemma, we consider the constant $c > 0$ using in \cref{theo:sync-expander} to bound the diameter of the big component.
\begin{lemma}\label{lemma:sync-view-honest-dispersal}
    If $f < n/(72 + c)$, then all but at most $7f$ honest parties will decide at most $c\log n$ rounds after getting a dispersal proof.
\end{lemma}

\begin{proof}
    We consider the creating of the first dispersal proof, which is happens before any party gets it. Because of \cref{coro:sync-view-commit-disp-same}, all dispersal proofs will be for the same value. If the dispersal proof is for a disagreement evidence, the party will immediately decide $\bot$ after getting it.

    Thus, from now on, we assume the dispersal proof is for an agreement evidence. Using \cref{lemma:sync-retrieval-parties-know}, at least $t+1-f$ honest parties know the value. Moreover, for the creation of the dispersal proof, at least $k$ parties must have given a threshold signature for it. Because $k = \lceil (n+t+1) / 2 \rceil$ and $t+1 \geq n/2$, then $k \geq 3n / 4$ so at most $n - 3n/4 \leq n/4$ parties did not send a partial signature for it. Hence, among the $t+1-f$ honest parties which know the value, at least $t + 1 - f - n/4$ sent a partial signature to the leader. We note that before sending the partial signature, a party will disperse in the expander $G$ the value associated with the commit proof it has if it knows it. Because of \cref{coro:sync-view-commit-disp-same}, the value for the commit proof is the same as the one for any dispersal proof. Thus, by the time the dispersal proof is created, at least $t + 1 - f - n/4$ honest parties started dispersing the decided value in $G$. Let us denote $Y$ the set of these parties.

    We assume $f < n/(72 + c)$, so $f < n/8$, hence $|Y| \geq t+1 -f - n/4 \geq n/2 - n/8 - n/4 \geq n/4$.
    We now apply \cref{theo:sync-expander} to the set of byzantine parties, whose size $f$ is less than $n/72$: There is a set $X$ of honest parties such that $|X| \geq n - 8f = n - n/9$ and the induced subgraph $G[X]$ is connected and has diameter at most $c\log n$. Because $|Y| > n/9$, we have $X \cap Y \ne \emptyset$, so there is an honest party $p \in Y$ which is also in $X$. Therefore, $p$ will start dispersing the decided value in $G[X]$ before any party gets a dispersal value. Because $G[X]$ has diameter at most $c \log n$, and parties relay the value they receive in $G$, all parties in $X$ will get the value at most $c\log n$ rounds later. So only honest parties not in $X$, which there are at most $(n - |X|) - f \leq (n - (n - 8f)) - f = 7f$ may not receive the value and decide in time.
\end{proof}

\begin{lemma}\label{lemma:sync-view-decide}
    If $f < min(n/4, n - k)$ and $2f + c\log n +4  \leq n$, then all parties decide within $\mathcal{O}(f + \log n)$ rounds within the leader-based agreement part.
\end{lemma}

\begin{proof}
    Using \cref{lemma:sync-view-get-dispersal}, within $2f+1$ views, all honest parties will get a dispersal proof. We remark that if a party did not decide $c \log n + 1$ rounds after receiving its dispersal proof, it will broadcast \textsc{RequestShare} messages asking every party for their share of the output. Moreover, because $f < n/4$, using \cref{lemma:sync-view-dispersal-partial}, at least $n/4$ parties will already know the value and send their share back, so parties will always be able to decode the output and decide it at most two rounds after requesting shares. Hence all honest parties will decide within $2f + 1$ views and $c \log n + 3$ additional rounds. Because we run the view-based protocol for $n$ views and a view last at least one round, if $2f + c \log n + 4 \leq n$, all honest parties will have decided before the end of the first part of the protocol and within $2f + c \log n + 4$ views so $\mathcal{O}(f + \log n)$ rounds.
\end{proof}

\begin{lemma}\label{lemma:sync-view-bit}
    \textsc{ViewByzantineAgreementSync} has $\mathcal{O}(n\cdot(L + f \cdot \kappa))$ bit complexity.
\end{lemma}

\begin{proof}
    We first remark that the retrieval protocol will return a non-$\bot$ evidence, while being called by an honest leader, at most once. The reason for this is that if an honest leader gets an agreement or disagreement evidence from the retrieval protocol, it will send it (or a lock or dispersal proof) to any subsequent leader which asks for it, which would cause any honest leader to skip running the retrieval protocol.

    We look at the bit complexity of using the expander: each party sends the decided $L$-bit value and proof at most once to its neighbors in the expander graph $G$. Because $G$ is regular with degree $D = \mathcal{O}(1)$, the total bit complexity for it is $\mathcal{O}(n(L + \kappa))$. We now look at the view-based protocol, outside of the retrieval protocol, the expander and asking for shares. We remark that in this form, it is the exact same approach as the original algorithm \textsc{ViewByzantineAgreement} from \cite{constantinescu2025few}, so using \cite[Corollary B.10, Lemma B.11]{constantinescu2025few}, we have $\mathcal{O}(n \cdot f)$ word complexity. Because each word has $\mathcal{O}(\kappa)$ bits, this results in a $\mathcal{O}(n \cdot f \cdot \kappa)$

    \textbf{If $f < n/(72 + c)$ and $2f + c\log n +4  \leq n$}: We remark that we have $f < min(n/4, n - k)$. We look at the bit complexity of each part. As stated previously, at most a single honest leader will run the retrieval protocol and get a non-$\bot$ value, which according to \cref{lemma:sync-leader-bits} has $\mathcal{O}(n(L + \kappa))$ bit complexity. Using \cref{lemma:sync-view-get-dispersal}, all honest parties will get a dispersal proof within $\mathcal{O}(f)$ views and thus stop calling the retrieval protocol. Because each call with a byzantine leader has bit complexity $\mathcal{O}(L + n \cdot \kappa)$, the total bit complexity is $\mathcal{O}(L + n \cdot \kappa)$.  Finally, using \cref{lemma:sync-view-honest-dispersal}, at most $7f$ honest parties (along with the $f$ byzantine parties) will ask for a share to other parties. Each of the $n-f$ honest parties may thus send back a share, of size $\mathcal{O}(L/n + \kappa)$ to each of the $\mathcal{O}(f)$ parties which requests it, resulting in a $n \cdot 8f \cdot \mathcal{O}(L/n + \kappa) = \mathcal{O}(n\cdot(L + f \cdot \kappa))$ bit complexity. Summing all these complexities together, we get the expected bit complexity.

    \textbf{If $f \geq n/(72 + c)$ or $2f + c\log n +4  > n$}: We note that in this case, $f = \Omega(n)$. As stated previously, at most a single honest leader will run the retrieval protocol and get a non-$\bot$ value, which according to \cref{lemma:sync-leader-bits} has $\mathcal{O}(n(L + \kappa))$ bit complexity. In the other views, each retrieval protocol invocation has bit complexity $\mathcal{O}(L + n \cdot \kappa)$ using \cref{lemma:sync-leader-bits} and happens at most $n$ times (once per view). So the total bit complexity for this part is $\mathcal{O}(n(L + n \cdot \kappa))$. Regarding asking for shares, every one of the $n$ parties may end up asking for a share to the $n$ parties. Each share taking $\mathcal{O}(L/n + \kappa)$ bits, the total is $\mathcal{O}(n(L + n\kappa))$ bits. Summing all these complexities together, we get  $\mathcal{O}(n(L + n \cdot \kappa))$ which is the expected complexity as $f = \Omega(n)$.
\end{proof}

We can now look at the fallback approach given in \textsc{SyncBA}:

\begin{lemma}\label{lemma:sync-qa-termination}
    All honest parties decide in $\mathcal{O}(f + \log n)$ rounds.
\end{lemma}

\begin{proof}
    If $f < min(n/4, n - k)$ and $2f + c\log n +4  \leq n$, then using \cref{lemma:sync-view-decide}, all honest parties decide in $\mathcal{O}(f + \log n)$ round before reaching the fallback.

    Otherwise, we have $f = \Omega(n)$. As the fallback protocol $\mathcal{A}_{fallback}$ from \cite{nayak2020extension} has $\mathcal{O}(n)$ round complexity, we remark that the whole protocol terminates in $\mathcal{O}(n)$ rounds and parties always decide by the end of the protocol if they have not yet. Thus the time to decide is $\mathcal{O}(n) = \mathcal{O}(f) = \mathcal{O}(f + \log n)$ in this case too.
\end{proof}

We remark that some parties may agree in the leader-based part of the protocol while some others may agree in the fallback part. Moreover, we accept as certificate for the decision either a dispersal proof or a certificate threshold after running the fallback. We show that we can ensure agreement and that all certificate are for the decided value:

\begin{lemma}\label{lemma:sync-agreement-on-dispersal}
    If at any point during (or after) the execution of the protocol, a dispersal proof exists associated with a value $v$, then all honest parties will decide $v$. Moreover, if a party needs to get a $CERTIFY$ threshold signature at the end of the algorithm, it will always be able to get the required $t+1$ partial signatures within one round. And no $CERTIFY$ signature will exist for values other than $v$.
\end{lemma}

\begin{proof}
    Assume at some point a dispersal proof exists associated with a value $v$. Using \cref{coro:sync-view-commit-disp-same}, all commit and dispersal proofs will be associated with $v$. Thus any party deciding in the leader-based part of the protocol will decide $v$, as it must have a dispersal proof associated with it. 

    We now consider the fallback part. Because a dispersal proof can exist, it means $k$ parties sent a partial signature for it during the leader-based protocol. Because $k > t$, at least one honest party $p$ sent a partial signature for it and to do so must have received a commit proof associated with $v$. Therefore, every honest party which has not yet decided by the fallback will send a help message and receive the commit (or dispersal) proof from $p$. Thus all honest parties will have a commit or dispersal proof for $v$ and thus try to reconstruct it. 

    We remark that a commit or dispersal proof is layered on top of an agreement or disagreement evidence. If it is a disagreement evidence, all parties will immediately get the value ($\bot$). If it is an agreement evidence, using \cref{lemma:sync-retrieval-parties-know}, at least $t + 1 - f \geq 1$ honest party $q$ knows $v$. Let $Q$ be the set of parties not knowing $v$. By design of the algorithm, every honest party already knowing $v$ will send their share as well as $u$'s share to a party $u \in Q$ when it asks for it. Parties in $Q$ will always be able to get their share from $q$. Then they will send their share to everyone in the following round. Thus every party in $Q$ will get $n-f \geq n/4$ distinct shares and be able to reconstruct $v$.

    We now consider two cases:
    \begin{itemize}
        \item If at least one honest party got a dispersal proof in the leader-based protocol, then it will send it to every other party when asked. So by design of the algorithm, every honest party will decide $v$, ignoring the output of the fallback protocol. Moreover, they will only possibly give a partial $CERTIFY$ signature for $v$, so it won't be possible to build a $CERTIFY$ threshold signature for another value.
        \item If no honest party got a dispersal proof in the leader-based protocol, this means no party decided in this part (as they would need a dispersal proof). Thus all $n-f\geq t+1$ honest parties will send a help message and will run the fallback protocol with the input $v$ they got. Because the fallback protocol satisfies strong unanimity, its output will be $v$. So all parties will decide $v$. Moreover, they also will all send a partial signature for $v$ (which matches their dispersal proof if they have one), so parties will be able to get a $CERTIFY$ threshold signature if they need to within one round. Moreover, it won't be possible to build a $CERTIFY$ threshold signature for another value.
    \end{itemize}
\end{proof}

\begin{lemma}\label{sync-fallback-no-dispersal}
    If no dispersal proof ever exists, all honest parties will run the fallback protocol and decide its output.
\end{lemma}

\begin{proof}
    If no dispersal proof ever exists, it means no party decided during the leader-based protocol (as they need a dispersal proof to do so). So they will all send a help message. Therefore, every honest party will receive at least $n-t \geq t+1$ help messages and run the fallback protocol. Because there is no dispersal proof, every honest party will decide the output of the fallback protocol and emit a partial signature for it.
\end{proof}

\begin{lemma}\label{sync-fallback-agreement}
    The protocol satisfies agreement.
\end{lemma}
\begin{proof}
    If a dispersal proof ever exists, this is a consequence of \cref{lemma:sync-agreement-on-dispersal}. Otherwise, using \cref{sync-fallback-no-dispersal}, every honest party will run the fallback protocol and decide on its output. Because the fallback protocol satisfies agreement, its output will also satisfy agreement.
\end{proof}

\begin{lemma}\label{lemma:sync-certify-prop}
    A $CERTIFY$ threshold signature can only exist for the decided value. Moreover, if a party needs to get a $CERTIFY$ threshold signature at the end of the algorithm, it will always be able to get the required $t+1$ partial signatures within one round.
\end{lemma}

\begin{proof}
    If a dispersal proof ever exists, this is a consequence of \cref{lemma:sync-agreement-on-dispersal}. Otherwise, using \cref{sync-fallback-no-dispersal}, every honest party will send a partial signature for its decided value. Because it satisfy agreement, it will always be possible to get the $t+1$ partial signatures and it will not be possible to get a threshold signature for any other value.
\end{proof}

\begin{lemma}\label{lemma:sync-prot-validity}
    The protocol satisfies strong unanimity.
\end{lemma}

\begin{proof}
    We note that the value associated with a commit or dispersal proof originates from agreement or disagreement evidence and therefore satisfies strong unanimity as proven in \Cref{lemma:sync-retrieval-parties-know,lemma:sync-retrieval-disagreement}.
    If a dispersal proof ever exists, this is therefore a consequence of \cref{lemma:sync-agreement-on-dispersal}. Otherwise, using \cref{sync-fallback-no-dispersal}, every honest party will join the fallback protocol with either its own input, which satisfies strong unanimity, or the value associated with a commit or dispersal proof, which also satisfies strong unanimity. Thus all honest parties will join the fallback protocol with a value that satisfies strong unanimity. Because the fallback protocol satisfies strong unanimity, the returned value will do so too.
\end{proof}

\begin{lemma}
    The protocol has $\mathcal{O}(n(L + f\cdot \kappa))$ bit complexity.
\end{lemma}

\begin{proof}
    We proved in \cref{lemma:sync-view-bit} that the leader-based protocol has $\mathcal{O}(n(L + f\cdot \kappa))$ bit complexity. We now have two cases:
    \begin{itemize}
        \item \textbf{$f < min(n/4, n - k)$ and $2f + c\log n +4  \leq n$}: then using \cref{lemma:sync-view-decide}, all honest parties decide within the view-based protocol, so none of them send a help message. Because no $t+1$ help message can be received, no honest party will run the fallback protocol, and they won't ask for shares either as they already have the value. The only additional messages come from each of the $f$ byzantine parties which can ask for $2$ shares and a dispersal proof to any honest party. This has total bit complexity $f \cdot (n-f) \cdot \mathcal{O}(L/n + \kappa) = \mathcal{O}(n(L + f\cdot \kappa))$.
        \item \textbf{$f \geq min(n/4, n - k)$ or $2f + c\log n +4  > n$}: In this case, $f = \Omega(n)$. Outside of the fallback protocol, each party may send up to $3$ shares and a constant amount of proofs to every other parties, this has total bit complexity $\mathcal{O}(n \cdot n \cdot (L/n + \kappa))$. The fallback protocol from \cite{nayak2020extension} has bit complexity $\mathcal{O}(nL + n^2 \kappa)$. In total, the protocol has bit complexity $\mathcal{O}(nL + n^2 \kappa) = \mathcal{O}(n(L + f \cdot \kappa))$ because $f = \Omega(n)$.
    \end{itemize}
\end{proof}

\begin{theorem}
    The protocol satisfies agreement, strong unanimity, Provability, partial termination in $\mathcal{O}(f)$ rounds, and termination in $\mathcal{O}(f + \log n)$ rounds.
\end{theorem}

\begin{proof}
    This is a consequence of \cref{lemma:sync-prot-validity,lemma:sync-certify-prop,sync-fallback-agreement,lemma:sync-qa-termination,lemma:sync-view-partial}.
\end{proof}

\section{Asynchronous Quorum Agreement}\label{app:asyncQA}

\subsection{Pseudocode}
Our Asynchronous Quorum Agreement protocol relies on an existing optimal asynchronous MVBA protcol $\Pi_{MVBA}$ (for example Dumbo-MVBA \cite{yuhan2020dumbo}) a a binary BA protocol $\Pi_{BA}$ (for example \cite{cachin2000random}). We assume that the resiliency is optimal (i.e $n = 3t+1$).

\begin{dianabox}{\textsc{AsyncQA}}
\algoHead{Asynchronous Quorum agreement protocol for a party $p_i$ with input $\vin$}
\begin{algorithmic}[1]
    \State $v \gets \Pi_{MVBA}(\vin)$
    \Statex

    \If{$v = \vin$}
        \State Send $(CHECK, OK)$ to all parties
    \Else{}
        \State Send $(CHECK, BAD)$ to all parties
    \EndIf
    \Statex

    \State Wait for $n-t$ check messages
    \If{Received $t+1$ $(CHECK, OK)$ messages}
        \State $b \gets 1$
    \Else{}
        \State $b \gets 0$
    \EndIf
    \State $d \gets \Pi_{BA}(b)$
    \Statex
    
    \If{d = 0}
        \State $v \gets \ast$
    \EndIf
    \State Send $tsign_{t+1}((CERT, v))$ to all parties
    \State Wait for $t+1$ partial signatures for $v$: $(\rho_p)$
    \State $cert \gets tcombine((CERT, v), (\rho_p))$
    \State Decide $(v, cert)$
\end{algorithmic}
\end{dianabox}
As an implementation detail, we note that a $1$-bit tag ($MVBA$ or $BA$) might have to be added to $\Pi_{BA}$ and $\Pi_{MVBA}$'s messages and signatures to prevent a message from $\Pi_{BA}$ being mistakenly used by another party lagging behind and still running $\Pi_{MVBA}$.

\subsection{Analysis}

\begin{lemma}
    \textsc{AsyncBA} has expected bit complexity $\mathcal{O}(n \cdot (L + t \cdot \kappa))$ and constant expected latency.
\end{lemma}

\begin{proof}
    \textsc{AsyncBA} consists of an invocation of $\Pi_{MVBA}$, an invocation of $\Pi_{BA}$ and two rounds of all-to-all broadcast. The all-to-all broadcasts have constant latency and $\mathcal{O}(n^2 \cdot \kappa) = \mathcal{O}(n \cdot t \cdot \kappa)$ bit complexity. The remaining part is the sum of the message and bit complexity of $\Pi_{MVBA}$, and $\Pi_{BA}$ \cite{yuhan2020dumbo,cachin2000random} which give the above complexities.
\end{proof}
We note that a consequence of the previous lemma is that the protocol satisfies probabilistic termination and thus partial probabilistic termination.

\begin{lemma}
    \textsc{AsyncBA} satisfies agreement, strong unanimity and provability.
\end{lemma}

\begin{proof}
    Regarding agreement, we note that the output of \textsc{AsyncBA} is a deterministic function of the output of $\Pi_{MVBA}$ and $\Pi_{BA}$, which both satisfy agreement. Thus the output is also the same for every honest party and satisfies agreement.
    For provability, the certificate for a value $v$ is a $(t+1)$-threshold certificate signed by parties which decide $v$. Given this certificate, it means there is at least one honest party which decided $v$ and so it was the decided value because of agreement.
    Strong unanimity: Assume all honest parties have the same input $v$. Using $\Pi_{MVBA}$'s strong unanimity, $v$ will be returned by it. So every honest party will broadcast a $(CHECK, OK)$ message. Given $n-t \geq 2t+1$ messages, at least $t+1$ of them will come from honest parties and therefore will be OKs. So all honest parties will have their bit $b$ set to $1$. Using $\Pi_{BA}$'s strong unanimity, $d = 1$ will be decided, so $v$ will be decided.
\end{proof}

\begin{lemma}
    \textsc{AsyncBA} satisfies representativity.
\end{lemma}
Assume a value $v \ne \ast$ is decided. This implies that the value $d$ decided was $1$. So using $\Pi_{BA}$'s strong unanimity, it implies that at least one honest party set its bit to $1$ and thus received at least $t+1$ $(CHECK, OK)$ messages. Because only byzantine parties or honest parties can send this message, we get the intended $f + n_v \geq t+1 \geq n/3$. 
\end{document}